\documentclass[sigconf]{aamas}

\usepackage[tbtags,fleqn]{mathtools}

\usepackage[ruled, linesnumbered, noend, noline]{algorithm2e}

\SetCommentSty{mycommfont}

\usepackage{multirow}

\usepackage{makecell}

\usepackage{array}

\theoremstyle{acmplain}
\newtheorem{theorem}{Theorem}[section]

\newtheorem{lemma}[theorem]{Lemma}
\newtheorem{claim}[theorem]{Claim}
\newtheorem{corollary}[theorem]{Corollary}
\theoremstyle{acmdefinition}

\newtheorem{definition}[theorem]{Definition}
\usepackage{balance} 

\doi{IPZQ7320}

\makeatletter
\gdef\@copyrightpermission{
  \begin{minipage}{0.2\columnwidth}
   \href{https://creativecommons.org/licenses/by/4.0/}{\includegraphics[width=0.90\textwidth]{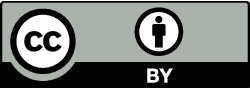}}
  \end{minipage}\hfill
  \begin{minipage}{0.8\columnwidth}
   \href{https://creativecommons.org/licenses/by/4.0/}{This work is licensed under a Creative Commons Attribution International 4.0 License.}
  \end{minipage}
  \vspace{5pt}
}
\makeatother

\setcopyright{ifaamas}
\acmConference[AAMAS '26]{Proc.\@ of the 25th International Conference
on Autonomous Agents and Multiagent Systems (AAMAS 2026)}{May 25 -- 29, 2026}
{Paphos, Cyprus}{C.~Amato, L.~Dennis, V.~Mascardi, J.~Thangarajah (eds.)}
\copyrightyear{2026}
\acmYear{2026}
\acmDOI{}
\acmPrice{}
\acmISBN{}

\title[A Radius-Sensitive Approximation Algorithm for Connected Submodular Maximization]{A Radius-Sensitive Approximation Algorithm for Connected Submodular Maximization}

\author{Philip Cervenjak}
\affiliation{
  \institution{School of Computing and Information Systems, The University of Melbourne}
  \city{Melbourne}
  \country{Australia}}
\email{pcervenjak@student.unimelb.edu.au}

\author{Junhao Gan}
\affiliation{
  \institution{School of Computing and Information Systems, The University of Melbourne}
  \city{Melbourne}
  \country{Australia}}
\email{junhao.gan@unimelb.edu.au}

\author{Naonori Kakimura}
\affiliation{
  \institution{Department of Mathematics, \\Keio University}
  \city{Yokohama}
  \country{Japan}}
\email{kakimura@math.keio.ac.jp}

\author{Seeun William Umboh}
\affiliation{
  \institution{School of Computing and Information Systems, The University of Melbourne \& ARC Training Centre in Optimisation Technologies, Integrated Methodologies, and Applications (OPTIMA)}
  \city{Melbourne}
  \country{Australia}
  }
\email{william.umboh@unimelb.edu.au}

\author{Anthony Wirth}
\affiliation{
  \institution{School of Computer Science, The University of Sydney \& School of Computing and Information Systems, The University of Melbourne\\[3mm]}
  \city{Sydney}
  \country{Australia}
  }
\email{anthony.wirth@sydney.edu.au}

\begin{abstract}
\textit{Connected Submodular Maximization} (\textbf{CSM}) is a graph problem with important applications to wireless network deployment, path planning, epidemic outbreaks, and cancer genome studies. In \textbf{CSM}, we are given a graph $G$, a non-negative monotone submodular function $f$ on subsets of the vertex set of $G$, and an integer $k$. The goal is to select a tree in $G$, with $k$ edges, whose vertex set maximizes $f$. We also study the more general Directed and Directed Rooted variants of \textbf{CSM} (\textbf{DCSM} and \textbf{DRCSM} respectively). In both variants, $G$ is directed and the solution must be an out-tree in $G$, with $k$ edges, whose vertex set maximizes $f$; \textbf{DRCSM} further specifies a vertex to be the root of the selected out-tree. For \textbf{CSM}, several previous works have proposed polynomial time approximation algorithms; the state-of-the-art polynomial time algorithm achieves a $\Omega(\frac{1}{\sqrt{k}})$-approximation. We can also parameterize the approximation factor by the radius of the optimal solution, denoted by $r$; the state-of-the-art polynomial time algorithm achieves a $\Omega(\frac{1}{r})$-approximation.

In this paper, we improve on the state-of-the-art approximation factor for \textbf{CSM} with respect to $r$ as well as $k$, noting that $r \leq k$. We propose a polynomial time framework that, for (Directed) \textbf{CSM}, achieves a $\Omega(\frac{\varepsilon^{3}}{{r}^{\varepsilon}})$-approximation for every constant $\varepsilon \in (0, 1]$. For \textbf{DRCSM}, our framework achieves a $\Omega(\frac{\delta \varepsilon^{3}}{{r}^{\varepsilon}})$-approximation that violates the size constraint by at most a factor of $1 + \delta$ for every $\delta \in [\frac{1}{k}, 1]$. A key component of our framework is \GR{}, an algorithm for \textbf{DRCSM} that outputs a bicriteria approximation, i.e., an approximate solution that violates the size constraint by at most some factor. \GR{} takes an algorithm with a bicriteria approximation factor in terms of $k$ and outputs a solution with the same bicriteria approximation factor (up to constants) in terms of $r$. Moreover, to use as a subroutine for \textbf{DRCSM}, we propose the algorithm \RAd{}, which achieves a $\frac{1}{d+1}$-approximation that violates the size constraint by at most a factor of $(d+1)^{2} k^{\frac{1}{d}}$. \RAd{} uses a recursive greedy strategy, with $d$ denoting the number of levels of recursion used. This enables the dependence on $\varepsilon$ in the approximation factors of our overall framework.
\end{abstract}



\newcommand{\BibTeX}{\rm B\kern-.05em{\sc i\kern-.025em b}\kern-.08em\TeX}

\DeclareMathOperator*{\argmax}{arg\,max}

\newcommand{\Funcf}{f}
\newcommand{\GroundSet}{V}

\newcommand{\GroundSetSize}{n}

\newcommand{\Univ}{X}

\newcommand{\OptValue}{\textnormal{OPT}}
\newcommand{\OptValuep}{\OptValue'}

\newcommand{\OptTree}{\SetT^{*}}

\newcommand{\OptCenter}{\Vertv}
\newcommand{\OptNumSub}{\Delta}
\newcommand{\TreeWeight}{s}
\newcommand{\OptSubInd}{j}
\newcommand{\BestOptSubInd}{j^{*}}
\newcommand{\RecLevel}{\ell}
\newcommand{\IterInd}{i}

\newcommand{\SetX}{X}
\newcommand{\SetY}{Y}
\newcommand{\SetS}{S}
\newcommand{\BestSolSub}{\hat{\SetS}}
\newcommand{\SetT}{T}
\newcommand{\SetP}{P}

\newcommand{\Error}{\varepsilon}
\newcommand{\Errorv}{\delta}

\DeclareMathOperator{\polylog}{polylog}

\newcommand{\Card}{k}
\newcommand{\Cardp}{\Card'}

\newcommand{\Budget}{B}

\newcommand{\Funcg}{g}
\newcommand{\Funcgp}{g'}

\newcommand{\Reals}{\mathbb{R}}

\newcommand{\Max}{\textnormal{max}}
\newcommand{\Min}{\textnormal{min}}

\newcommand{\NumOptVertices}{k'}
\newcommand{\NumTerm}{z}
\newcommand{\NumTermp}{z'}

\newcommand{\HHop}{h}
\newcommand{\Graph}{G}

\newcommand{\VertexSet}{\GroundSet}
\newcommand{\EdgeSet}{E}
\newcommand{\NumVertices}{\GroundSetSize}

\newcommand{\NumEdges}{m}
\DeclareMathOperator{\Dist}{dist}

\newcommand{\Vertv}{v}
\newcommand{\Vertw}{w}

\newcommand{\Root}{\Vertv}

\newcommand{\OptVert}{w^{*}}

\newcommand{\Degree}{\Delta'}

\newcommand{\Ecc}{\epsilon}
\newcommand{\Radius}{r}
\newcommand{\OptRadius}{\Radius}
\newcommand{\Height}{\Radius}
\newcommand{\OptHeight}{\Height}

\newcommand{\NumSubtrees}{t}

\newcommand{\GR}{\textsc{GreedyRadius}}

\newcommand{\SetW}{W}

\newcommand{\SetTp}{\SetT'}
\newcommand{\Bestw}{\hat{\Vertw}_{\IterInd}}
\newcommand{\BestOptSub}{\hat{\SetT}_{\IterInd}}
\newcommand{\BestOptSuba}{\hat{\SetT}}

\newcommand{\Bestwa}{\hat{\Vertw}}

\newcommand{\Alg}{\textsc{Alg}}
\newcommand{\AppFactor}{\alpha}
\newcommand{\VioFactor}{\beta}
\newcommand{\Time}{\Gamma}

\newcommand{\TargetSubsize}{c}
\newcommand{\BestSubsize}{\hat{c}_{\IterInd}}
\newcommand{\BestSubsizea}{\hat{c}}
\newcommand{\TargetSubsizeMax}{\TargetSubsize_{\Max(\IterInd)}}

\newcommand{\TargetSubsizeMin}{\TargetSubsize_{\Min(\IterInd)}}

\newcommand{\TargetDiv}{q}
\newcommand{\List}{\mathcal{L}}
\newcommand{\SolSizeFunc}{M}

\newcommand{\RA}{\textsc{RecApprox}}

\newcommand{\RecDepth}{d}
\newcommand{\RAd}{\RA-$\RecDepth$}

\newcommand{\TargetSize}{b}

\newcommand{\Covf}{\Funcf(\VertexSet)}

\newcommand{\CSM}{Connected Submodular Maximization}
\newcommand{\CSMa}{\textbf{CSM}}

\newcommand{\DCSM}{Directed Connected Submodular Maximization}
\newcommand{\DCSMa}{\textbf{DCSM}}

\newcommand{\DRCSM}{Directed Rooted Connected Submodular Maximization}
\newcommand{\DRCSMa}{\textbf{DRCSM}}

\newcommand{\STOa}{\textbf{STO}}

\newcommand{\SO}{Submodular Orienteering}

\newcommand{\DRSTa}{\textbf{DRST}}

\newcommand{\DUSTa}{\textbf{DUST}}

\newcommand{\MCSBa}{\textbf{MCSB}}

\newcommand{\CMC}{Connected Maximum Coverage}
\newcommand{\CMCa}{\textbf{CMC}}

\newcommand{\DRCMC}{Directed Rooted Connected Maximum Coverage}
\newcommand{\DRCMCa}{\textbf{DRCMC}}

\newcommand{\CBCa}{\textbf{CBC}}

\newcommand{\DCBCa}{\textbf{DCBC}}

\newcommand{\BCDS}{Budgeted Connected Dominating Set}

\newcommand{\GSTa}{\textbf{GST}}

\newcommand{\DST}{Directed Steiner Tree}
\newcommand{\DSTa}{\textbf{DST}}

\newcommand{\PST}{Polymatroid Steiner Tree}
\newcommand{\PSTa}{\textbf{PST}}

\newcommand{\PDST}{Polymatroid Directed Steiner Tree}
\newcommand{\PDSTa}{\textbf{PDST}}

\newcommand{\Consts}{s}

\newcommand{\p}[1]{}

\begin{document}
\emergencystretch 3em


\pagestyle{fancy}
\fancyhead{}


\maketitle 


\section{Introduction}
We study the problem of \emph{\CSM{}} (\CSMa{}). In this problem, we are given an undirected graph $\Graph = (\VertexSet, \EdgeSet)$ with $\NumVertices$ vertices and $\NumEdges$ edges, a non-negative monotone submodular set function~$\Funcf$ whose ground set is~$\VertexSet$, and an integer~$\Card \geq 1$; the goal is to select a tree $\SetS \subseteq \Graph$, with~$\Card$ edges, whose vertex set maximizes~$\Funcf$.\footnote{Our formulation of \CSMa{} is equivalent to the formulation where the goal is to select a set $\SetS' \subseteq \VertexSet$, of $\Cardp$ vertices that are connected in $\Graph$, that maximizes $\Funcf$. This is because $\SetS'$ must have a spanning tree with $\Cardp-1 = \Card$ edges.} We also study \emph{\DCSM{}} (\DCSMa{}), which generalizes \CSMa{} by letting $\Graph$ be a directed graph; the goal is to select an out-tree (i.e., arborescence) $\SetS \subseteq \Graph$, with~$\Card$ edges, whose vertex set maximizes~$\Funcf$. We further study \emph{\DRCSM{}} (\DRCSMa{}), which generalizes \DCSMa{} by specifying a vertex $\Vertv$ to be the root of the selected out-tree.

\CSMa{} has received significant attention owing to its important applications. These include deploying a connected network, of limited size, with maximum coverage or throughput; specifically positioning a connected network of wireless routers or relays~\cite{kuo2014maximizing, Gao2018}, unmanned aerial vehicles~\cite{Xu2021, Xu2022, Li2023, Lv2024, wang2024approximation}, or wireless power chargers~\cite{Yu2019, zhou2021design, You2023}.
Another application of \CSMa{} is identifying a limited number of geographically connected regions that, if undervaccinated, are most susceptible to an epidemic outbreak \cite{Cadena2020}. This is motivated by the need for public health interventions, which are most effective within a small number of localized undervaccinated regions. Further, \CSMa{} is related to maintaining a communication network between robots in multi-robot task planning \cite{Shi2021}.

An important special case of \CSMa{} is \emph{\CMC{}} (\CMCa{}). Here, we are given a universe of elements $\Univ$, an undirected graph $\Graph$ where each vertex is a subset of $\Univ$, and an integer $\Card \geq 1$; the goal is to select a tree $\SetS \subseteq \Graph$, with $\Card$ edges, whose vertex set maximizes the number of its covered elements. \CMCa{} has application to deploying a wireless sensor network, of limited size, with maximum coverage \cite{Khuller2014, Huang2015, Huang2019}. \CMCa{} also has application to identifying genetic mutations associated with cancers~\cite{Vandin2011, Hochbaum2020}. In this application, we are given a universe of cancer patients and a gene interaction network, wherein each genetic mutation is associated with the subset of patients that have the mutation. As cancer is widely assumed to be caused by a network of mutated genes (called a \emph{pathway}), the goal is to select a limited set of connected gene mutations that `explain' the most cases of cancer (i.e., cover the most patients). \CMCa{} is also related to the Watchman Route Problem, where we are given a `map' of various locations, and the goal is to find a minimum-length path in the map such that every location is visible from somewhere along the path \cite{skyler2022solving, Livne2023}. Further, \DRCMC{} (\DRCMCa{}), i.e., the special case of \DRCSMa{} with a coverage objective, is related to problems where we are given a network that models interactions between agents, and the goal is to select a minimum-size rooted out-tree that reconstructs the propagation of some activity throughout the network, such as an epidemic outbreak~\cite{Rozenshtein2016, Mishra2023}.

\paragraph*{Approximation Algorithms.}
Most of the previous works on \CSMa{} have focused on \emph{polynomial time approximation algorithms}. \CSMa{} was first studied by Kuo et al. \cite{kuo2014maximizing}, where they showed a polynomial time $\frac{1-1/e}{5(\sqrt{\Card+1}+1)}$-approximation algorithm. Xu et al.~\cite{Xu2021} improved the approximation factor to $\frac{1-1/e}{\lfloor \sqrt{\Card+1} \rfloor}$, and Li et al. \cite{Li2023} further improved it to $\Omega( \sqrt{ \frac{\Consts}{\Card} })$, where $\Consts$ is a chosen integer parameter (their algorithm runs in polynomial time if $\Consts$ is set constant). Note that these works have only found constant-factor improvements in approximation. There also exists, for any constant $\Error > 0$, a polynomial time $\Omega(\frac{1}{\log^{2+\Error} \NumVertices})$-approximation algorithm for \CSMa{} as implied by Theorem 3.1 of Im et al. \cite{im2016minimum}. This, in turn, is implied by an algorithm for \PST{} due to Calinescu and Zelikovsky \cite{Calinescu2005}. However, this algorithm crucially relies on taking a tree embedding of the input graph, which does not work for general directed graphs.

We also mention that a number of works have studied budgeted variants of (Directed Rooted) \CSMa{} \cite{kuo2014maximizing, Ghuge2022, DAngelo2022} and (Directed Rooted) \CMCa{} \cite{Huang2015, Ran2016, Huang2019, DAngelo2022, DAngelo2025}, i.e., with edge or vertex costs. We outline these works in the related work (Section~\ref{sec:related work}), though we mention the state-of-the-art works on budgeted (Directed Rooted) \CSMa{} here.
Ghuge and Nagarajan \cite{Ghuge2022} studied \DRCSMa{} with edge costs, showing a quasi-polynomial time $\Omega(\frac{\log \log \NumOptVertices}{\log \NumOptVertices})$-approximation algorithm, where $\NumOptVertices$ is the number of vertices in an optimal tree. D'Angelo et al. \cite{DAngelo2022} studied (Directed) \CSMa{} and \DRCSMa{} with vertex costs, showing polynomial time algorithms: letting $\Budget$ denote the budget, they respectively gave a $\Omega(\frac{1}{\sqrt{\Budget}})$-approximation algorithm and, for every $\Errorv \in (0, 1]$, a $\Omega(\frac{\Errorv^{3}}{\sqrt{\Budget}})$-approximation algorithm that violates the budget by at most a factor of $1+\Errorv$.

Despite the aforementioned approximation algorithms, there appear to be no interesting approximation hardness results for \CSMa{}. The most we can say is that it is NP-hard to find an approximation better than $1-\frac{1}{e}$~\cite{feige1998threshold}, since the cardinality constrained problem is a special case (by taking a complete graph as input). However, for budgeted variants of \CSMa{}, there are near-logarithmic approximation hardness results \cite{Halperin2003, kortsarz2011approximating, Grandoni2019}. We outline these hardness results in the related work (Section~\ref{sec:related work}).

\paragraph*{Radius-Sensitive Approximation Algorithms.}
Previous works have also considered polynomial time algorithms for \CSMa{} whose approximation factors depend on one or more parameters that restrict the input graph,~$\Graph$, or the input function,~$\Funcf$. These algorithms are useful as, for certain parameter values, they can outperform algorithms whose approximation factors depend on $\NumVertices$ or $\Card$. In this paper, we are interested in the \emph{radius} of the optimal (out-)tree, denoted by $\OptRadius$. Other parameters that have been considered in the approximation factor include the doubling dimension of a metric graph \cite{Cadena2020}, the curvature of $\Funcf$ \cite{Lv2024}, and the $\HHop$-hop independence of $\Funcf$ \cite{Xu2022, Lv2024}.

For an instance of \CSMa{}, we define the radius, $\OptRadius$, of an optimal tree. First, let $\OptCenter$ be the \emph{center} of the optimal tree, i.e., $\OptCenter$ is the vertex that minimizes its maximum shortest distance to any other vertex in the optimal tree. Then $\OptRadius$ is this maximum distance. For an instance of \DCSMa{} or \DRCSMa{}, we define the radius, $\OptRadius$, of an optimal out-tree to be equivalent to its height, i.e., the number of edges in the longest directed path from its root. Note that $\OptRadius \leq \lceil \frac{\Card}{2} \rceil$ holds for an undirected tree and $\OptRadius \leq \Card$ for a directed out-tree.

The optimal solution radius~$\OptRadius$ is a natural parameter to consider as a smaller~$\OptRadius$ makes the vertices of an optimal (out-)tree more reachable from its center~$\OptCenter$, thus enabling better approximations. To see this intuitively, suppose~$\OptRadius = 1$; in this case, we can guess $\OptCenter$, initialize our solution with it, and then run the standard greedy algorithm on the (out-)neighbors of~$\OptCenter$. This achieves a $(1-\frac{1}{e})$-approximation, the same as for the cardinality constrained variant~\cite{nemhauser1978analysis}. Moreover, the optimal solution radius is expected to be small in `small-world' graphs, which are graphs where the average shortest distance between a pair of vertices is $O(\log \GroundSetSize)$.

Although the optimal solution radius $\OptRadius$ is a natural parameter, only two previous works have studied it in the approximation factor for \CSMa{} or the special case \CMC{} (\CMCa{}), both achieving $\Omega(\frac{1}{\OptRadius})$-approximations. Vandin et al. \cite{Vandin2011} first showed a polynomial time $\frac{e-1}{(2e-1)\OptRadius}$-approximation algorithm for \CMCa{}. Hochbaum and Rao \cite{Hochbaum2020} improved this approximation by a constant factor to $\max\{ (1 - \frac{1}{e} ) ( \frac{1}{\OptRadius} - \frac{1}{\Card+1} ), \frac{1}{\Card+1} \}$ and extended their result to \CSMa{}. However, observe that when $\OptRadius = \Theta(\Card)$, these algorithms only achieve a $\Omega(\frac{1}{\Card})$-approximation, the same as the trivial approximation achieved by selecting the most valuable vertex in $\Graph$. Thus, it would be ideal to achieve a unified approximation with the same dependence on $\OptRadius$ and $\Card$, asymptotically speaking.
 
\subsection{Our Contributions} \label{sec:contr}
Our main result is Theorem~\ref{thm:GR+RA for DCSM} below.
\begin{theorem} \label{thm:GR+RA for DCSM}
    Let $(\Graph, \Funcf, \Card)$ be an instance of (Directed) \CSMa{}. Then, for every $\Error \in (0, 1]$, there exists an $\frac{\Error^3}{16 (1+2\Error)^{3} {\OptRadius}^{\Error}}$-approximation algorithm for (Directed) \CSMa{} that runs in time $O(\Card \NumVertices^{\lceil\frac{1}{\Error}\rceil+2} \OptRadius^{2\lceil\frac{1}{\Error}\rceil+1})$.
\end{theorem}

Theorem~\ref{thm:GR+RA for DCSM} shows that, for every constant $\Error < \frac{1}{2}$, we improve on the previous polynomial time $\Omega(\frac{1}{\sqrt{\Card}})$-approximations \cite{kuo2014maximizing, Xu2021, Li2023} (recalling that $\OptRadius \leq \Card$). Moreover, for every constant $\Error < 1$, we improve on the previous polynomial time $\Omega(\frac{1}{\OptRadius})$-approximations \cite{Vandin2011, Hochbaum2020}.

We also obtain Theorem~\ref{thm:GR+RA for DRCSM} for the more general problem of \DRCSM{} (\DRCSMa{}). This theorem gives a bicriteria $(\AppFactor, \VioFactor)$-approximation algorithm, which is an algorithm that outputs an $\AppFactor$-approximate solution that may violate the size constraint by at most a factor of $\VioFactor \geq 1$.
\begin{theorem} \label{thm:GR+RA for DRCSM}
    Let $(\Graph, \Funcf, \Card, \Vertv)$ be an instance of \DRCSMa{}. Then, for every $\Error \in (0, 1]$ and every $\Errorv \in [\frac{1}{\Card}, 1]$, there exists a bicriteria $( \frac{\Errorv \Error^3}{16 (1+2\Error)^{3} {\OptRadius}^{\Error}}, 1+\Errorv)$-approximation algorithm for \DRCSMa{} that runs in time $O(\Card \NumVertices^{\lceil\frac{1}{\Error}\rceil+2} \OptRadius^{2\lceil\frac{1}{\Error}\rceil+1})$.
\end{theorem}

\paragraph*{Our Algorithmic Framework.}
Our framework for achieving Theorems~\ref{thm:GR+RA for DCSM} and \ref{thm:GR+RA for DRCSM} involves combining two novel algorithms for \DRCSMa{}, namely \GR{} and \RAd{}, whose guarantees we outline below.

\GR{} takes a bicriteria $(\AppFactor(\Card), \VioFactor(\Card))$-approximation algorithm for \DRCSMa{} that runs in time $\Time(\NumVertices, \Card)$ and outputs a bicriteria $( \frac{1}{2} \AppFactor(\OptRadius), 4\VioFactor(\OptRadius) )$-approximation in time $O(\frac{\Card \NumVertices}{\OptRadius} \Time(\NumVertices, \OptRadius))$; we formally state these guarantees in Theorem~\ref{thm:GR}. 
By selecting a valuable size-$\Card$ out-subtree of this solution, we obtain a feasible $\frac{\AppFactor(\OptRadius)}{16\VioFactor(\OptRadius)}$-approximation for (Directed) \CSMa{}; we formally state this result in Corollary~\ref{cor:GR for DCSM}. Alternatively, by applying a simple trimming process, we obtain a bicriteria $(\frac{\Errorv \AppFactor(\OptRadius)}{16 \VioFactor(\OptRadius)}, 1+\Errorv)$-approximation for \DRCSMa{}; we formally state this result in Corollary~\ref{cor:GR for DRCSM}.

We further propose \RAd{} as an algorithm that can be used as a subroutine in our framework. For every integer $\RecDepth \geq 1$, \RAd{} is a bicriteria $(\frac{1}{\RecDepth+1}, (\RecDepth+1)^{2} \Card^{\frac{1}{\RecDepth}})$-approximation algorithm for \DRCSMa{} that runs in time $O(\NumVertices^{\RecDepth+1} \Card^{2\RecDepth+2})$. We formally state the performance guarantees of \RAd{} in Theorem~\ref{thm:RAd}.

By running \GR{} with subroutine \RAd{}, we obtain a bicriteria $(\frac{1}{2(\RecDepth+1)}, 4(\RecDepth+1)^{2} \OptRadius^{\frac{1}{\RecDepth}})$-approximation algorithm for \DRCSMa{} that runs in time $O(\Card \NumVertices^{\RecDepth+2} \OptRadius^{2\RecDepth+1})$. Then, by selecting a size-$\Card$ out-subtree (Corollary~\ref{cor:GR for DCSM}), we obtain a feasible $\frac{1}{16(\RecDepth+1)^{3} \OptRadius^{1/\RecDepth}}$-approximation for (Directed) \CSMa{}. Otherwise, by the trimming process (Corollary~\ref{cor:GR for DRCSM}), we obtain a bicriteria $(\frac{\Errorv}{16 (\RecDepth+1)^{3} \OptRadius^{1/\RecDepth}}, 1+\Errorv)$-approximation for \DRCSMa{}. We assign $\RecDepth = \lceil \frac{1}{\Error}\rceil$ in the above approximation factors and use the bounds $\frac{1}{\Error} \leq \lceil \frac{1}{\Error} \rceil \leq \frac{1}{\Error} + 1$ to derive Theorems~\ref{thm:GR+RA for DCSM} and \ref{thm:GR+RA for DRCSM} respectively.

In general, we can plug any algorithm for \DRCSMa{} into our framework to convert its (bicriteria) approximation factor's dependence on $\Card$ to $\OptRadius$. For example, we can plug in the quasi-polynomial time
$\Omega(\frac{\log \log \Card}{\log \Card})$-approximation algorithm by Ghuge and Nagarajan \cite{Ghuge2022} to achieve a bicriteria $(\Omega(\frac{\log \log \OptRadius}{\log \OptRadius}), 4)$-approximation.

\subsection{Technical Overview}
Here we outline the main ideas used in \GR{} and \RAd{}, while comparing with previous approaches. Both of our algorithms aim to construct an out-tree for \DRCSMa{} by greedily combining valuable out-subtrees, rather than valuable vertices, along with connecting paths from the root vertex.

\paragraph*{\GR{}.}
We first explain why the previous-best algorithm~\cite{Hochbaum2020} only achieves a $\Omega(\frac{1}{\OptRadius})$-approximation for \CSMa{}.
This algorithm first initializes a solution,~$\SetS$, with the root~$\Root$, and then uses the following greedy approach. While $\SetS$ can be feasibly updated: (1) find a vertex,~$\Bestwa$, of distance at most~$\OptRadius$ from~$\Root$, whose shortest path from~$\Root$ has maximum marginal gain to $\SetS$, and (2) add this shortest $\Root$--$\Bestwa$ path to~$\SetS$. The issue here is that, for a given~$\Bestwa$, there may be many shortest $\Root$--$\Bestwa$ paths (which cannot all be considered in polynomial time).
This means it is possible for the algorithm to just add $\Root$--$\Bestwa$ paths in which the only valuable vertex is $\Bestwa$. Thus, it may incur up to $\OptRadius$ edges to add each valuable vertex to the solution.

To improve on the above approach, \GR{} constructs a solution $\SetS$ by greedily adding valuable out-subtrees. \GR{} constructs each added out-subtree by guessing the root, $\Bestwa$, of a valuable size-$\OptRadius$ out-subtree, $\BestOptSuba$, of the optimal out-tree; then it calls the given $(\AppFactor(\Card), \VioFactor(\Card))$-approximation subroutine to solve a sub-instance of \DRCSMa{} with size constraint $\OptRadius$ and the guess of root $\Bestwa$. The subroutine must output an $(\AppFactor(\OptRadius), \VioFactor(\OptRadius))$-approximation of $\BestOptSuba$. By combining sufficiently many of these out-subtrees, along with connecting paths of length at most $\OptRadius-1$, the solution $\SetS$ is a $(\frac{1}{2}\AppFactor(\OptRadius), 4\VioFactor(\OptRadius))$-approximate out-tree.

\paragraph*{\texorpdfstring{\RAd{}}{\RA-\textit{d}}.}
\RAd{} is a procedure that initializes our main recursive algorithm, \RA{}, so that the maximum depth of its recursion tree (counted by `edges') is $\RecDepth$; we call $\RecDepth$ the \emph{recursion depth} of \RAd{}. Note that we use $\TargetSize$ to denote the size constraint of a sub-instance of \DRCSMa{} to distinguish it from the size constraint, $\Card$, of a main instance of \DRCSMa{}.

\RA{} uses a recursive greedy strategy that generalizes the strategy by Kuo et al.~\cite{kuo2014maximizing} for implicitly achieving a bicriteria $(1 - \frac{1}{e}, \TargetSize)$-approximation for \DRCSMa{} (which they used to achieve an overall $\Omega(\frac{1}{\sqrt{\Card}})$-approximation for \CSMa{}).
We explain the basic approach of Kuo et al.~\cite{kuo2014maximizing} first: given an instance of \DRCSMa{} with size constraint $\TargetSize$ and root $\Root$, initialize a solution, $\SetS$, with the root $\Root$. Then greedily add $\TargetSize$ vertices, of distance at most $\TargetSize$ from $\Vertv$, to $\SetS$, along with connecting paths from $\Vertv$.

To implement a recursive greedy strategy, \RA{} takes in its input a value $\TargetDiv > 1$; this is to reduce the size constraint to at most $\frac{\TargetSize}{\TargetDiv}$ in the recursive calls. Then \RA{} constructs a solution $\SetS$ by greedily adding valuable out-subtrees. \RA{} constructs each added out-subtree by guessing the size, $\BestSubsizea \leq \frac{\TargetSize}{\TargetDiv}$, and root, $\Bestwa$, of a valuable out-subtree, $\BestOptSuba$, of the optimal out-tree; then it makes a recursive call to solve a sub-instance of \DRCSMa{} with the guesses of size $\BestSubsizea$ and root $\Bestwa$. \RA{} continues adding valuable out-subtrees until the sum of their corresponding size-guesses is $\TargetSize$. \RA{} also adds connecting paths of length at most $\TargetSize - 1$.

Lastly, to achieve the bicriteria $(\frac{1}{\RecDepth+1}, (\RecDepth+1)^2 \Card^{\frac{1}{\RecDepth}})$-approximation for \DRCSMa{}, \RAd{} calls \RA{} with size constraint $\Card$, root $\Vertv$, and value $\TargetDiv = \Card^{\frac{1}{\RecDepth}}$, and returns the solution output by \RA{}. We can see that the assignment of $\TargetDiv = \Card^{\frac{1}{\RecDepth}}$ ensures a recursion depth of $\RecDepth$ since it makes the size constraints in a chain of recursive calls at most $\Card^{\frac{\RecDepth}{\RecDepth}}, \Card^{\frac{\RecDepth-1}{\RecDepth}}, \dots, \Card^{\frac{1}{\RecDepth}}, \Card^{\frac{0}{\RecDepth}}$.

We point out that \RA{} uses a similar recursion to an existing algorithm by Calinescu and Zelikovsky \cite{Calinescu2005} for the problem of \PDST{}. Using additional ideas, it is possible to adapt their algorithm to achieve a bicriteria approximation for \DRCSMa{} similar to that of \RA{}, but the resulting running time is worse than that of \RA{} by a polynomial factor. We compare our algorithm with theirs and sketch how to adapt it for \DRCSMa{} in the related work (Section~\ref{sec:related work}).

\subsection{Paper Structure}
We present preliminaries in Section~\ref{sec:prelim}, \GR{} in Section~\ref{sec:poly time approx algorithm}, \RA{} and \RAd{} in Section~\ref{sec:rec greedy approx algorithm}, related work in Section~\ref{sec:related work}, and conclusions in Section~\ref{sec:con}.

\section{Preliminaries} \label{sec:prelim}
In this paper, $\Graph$ denotes an undirected or directed graph, $\VertexSet(\Graph)$ denotes the vertex set of $\Graph$, and $\EdgeSet(\Graph)$ denotes the edge set of $\Graph$. Define functions $\NumVertices(\Graph) \coloneqq |\VertexSet(\Graph)|$, which is the \emph{cardinality} of $\Graph$, and $\NumEdges(\Graph) \coloneqq |\EdgeSet(\Graph)|$, which is the \emph{size} of $\Graph$. If $\Graph$ refers to the input graph of our problem, we simply use $\VertexSet$, $\EdgeSet$, $\NumVertices$, and $\NumEdges$ to denote the vertex set, edge set, cardinality, and size of $\Graph$ respectively.

Given an undirected or directed graph $\Graph$, the \emph{distance} from vertex $\Vertv$ to vertex $\Vertw$, $\Dist_{\Graph}(\Vertv, \Vertw)$, is the number of edges in a shortest (directed) path from $\Vertv$ to $\Vertw$. If $\Graph$ refers to the graph in the problem input, we may simply use $\Dist(\Vertv, \Vertw)$ to denote the distance from $\Vertv$ to $\Vertw$.

\paragraph*{Undirected Graph Terms.}
Given an undirected graph $\SetT$, the \emph{eccentricity} of a vertex $\Vertv \in \VertexSet(\SetT)$, $\Ecc_{\SetT}(\Vertv)$, is the maximum distance from $\Vertv$ to $\Vertw$ over all $\Vertw \in \VertexSet(\SetT)$; that is, $\Ecc_{\SetT}(\Vertv) \coloneqq \max_{\Vertw \in \VertexSet(\SetT)} \Dist_{\SetT}(\Vertv, \Vertw)$.

The \emph{radius}, $\Radius$, of $\SetT$ is the minimum eccentricity over all $\Vertv \in \VertexSet(\SetT)$; that is, $\Radius = \min_{\Vertv \in \VertexSet(\SetT)} \Ecc_{\SetT}(\Vertv)$. A \emph{center} of $\SetT$ is a vertex $\Vertv \in \VertexSet(\SetT)$ whose eccentricity is equal to the radius (there can be more than one center).

\paragraph*{Directed Graph Terms.}
A directed graph $\SetT$ is an \emph{out-tree} (i.e., \emph{arborescence}) if there exists a vertex $\Vertv \in \VertexSet(\SetT)$ such that for every $\Vertw \in \VertexSet(\SetT)$, there is exactly one directed path in $\SetT$ from $\Vertv$ to $\Vertw$. The vertex $\Vertv$ is called the \emph{root} of $\SetT$.

The \emph{radius} (i.e., \emph{height}), $\Height$, of an out-tree $\SetT$ is the maximum distance from its root $\Vertv$ to $\Vertw$ over all $\Vertw \in \VertexSet(\SetT)$; that is, $\Height = \max_{\Vertw \in \VertexSet(\SetT)} \Dist_{\SetT}(\Vertv, \Vertw)$.

\paragraph*{Submodular Functions.}
Let $\Funcf$ be a set function whose ground set is $\VertexSet$. Given subsets $\SetX, \SetY \subseteq \VertexSet$, let $\Funcf(\SetX \mid \SetY) \coloneqq \Funcf(\SetX \cup \SetY) - \Funcf(\SetY)$, which is called the \emph{marginal gain} of $\SetX$ to $\SetY$. As an abuse of notation, we may apply $\Funcf$ to a subgraph $\SetS \subseteq \Graph$, which means applying it to the vertex set of $\SetS$. Also, we may apply $\Funcf$ to a single vertex $\Vertv \in \VertexSet$, which means applying it to the set $\{ \Vertv \}$.

We consider functions $\Funcf$ that are non-negative monotone submodular, as defined below.

\begin{definition}[Submodular function] \label{def:sub func}
    A set function $\Funcf : 2^{\GroundSet} \rightarrow \Reals$ is \emph{submodular} iff for all $\SetX, \SetY \subseteq \GroundSet$ where $\SetX \subseteq \SetY$, and for all $\Vertv \in \GroundSet \setminus \SetY,\, \Funcf( \Vertv \mid \SetX) \geq \Funcf( \Vertv \mid \SetY)$.
\end{definition}

A set function $\Funcf$ is \emph{monotone} iff for all $\SetX, \SetY \subseteq \GroundSet$ where $\SetX \subseteq \SetY,\, \Funcf(\SetX) \leq \Funcf(\SetY)$; and $\Funcf$ is \emph{non-negative} iff for all $\SetX \subseteq \GroundSet,\, \Funcf(\SetX) \geq 0$.

We make the standard assumption that $\Funcf$ (or the marginal gain function) is queried via a \emph{value oracle}; we specifically assume a \emph{strong} value oracle, which allows both feasible and infeasible sets to be queried. For convenience, we analyze the running time of each algorithm by only counting the number of value oracle queries it makes.

\paragraph*{Problem Definitions.}
The problem of \emph{\CSM{}} (\CSMa{}) is defined as follows. We are given an undirected graph $\Graph = (\VertexSet, \EdgeSet)$ with $\NumVertices$ vertices and $\NumEdges$ edges, a non-negative monotone submodular function $\Funcf : 2^\VertexSet \rightarrow \Reals_{\geq 0} $, and an integer $\Card \geq 1$. The goal is to select a tree $\SetS \subseteq \Graph$, with $\NumEdges(\SetS) \leq \Card$, that maximizes $\Funcf$.

The problem of \emph{\DCSM} (\DCSMa{}) generalizes \CSMa{} by letting $\Graph$ be a directed graph and additionally requiring that the selected $\SetS \subseteq \Graph$ is an out-tree. The problem of \emph{\DRCSM{}} (\DRCSMa{}) in turn generalizes \DCSMa{} by specifying a vertex $\Vertv$ to be the root of the selected out-tree $\SetS \subseteq \Graph$.

Note that we use $\Card$ to constrain the \emph{size}, i.e., number of edges, of $\SetS$. This is convenient for us since our analyses heavily rely on partitioning trees into edge-disjoint (out-)subtrees and constructing solutions out of edge-disjoint (out-)subtrees.

We use $\OptTree$ to denote an optimal (out-)tree to any one of our problems, and  $\Vertv$ and $\OptHeight$ to denote its root and radius (height) respectively.

\paragraph*{Bicriteria Approximation.}
Let $\AppFactor \in (0, 1]$ and $\VioFactor \geq 1$ be values, which we call the \emph{approximation factor} and the \emph{violation factor} respectively. Then, for any given instance of one of our problems, an \emph{$(\AppFactor, \VioFactor)$-approximation} is a solution $\SetS$ with value $\Funcf(\SetS) \geq \AppFactor \Funcf(\OptTree)$ and size $\NumEdges(\SetS) \leq \VioFactor \Card$. Further, for any one of our problems, an \emph{$(\AppFactor, \VioFactor)$-approximation algorithm} returns an $(\AppFactor, \VioFactor)$-approximation for every instance of the problem.

\paragraph*{Tree Partitioning.}
Throughout our analyses, we use Lemma~\ref{lem:tree decomp} below, which is a simplified version of Lemma 2 of Khani and Salavatipour \cite{Khani2014}.

\begin{lemma}[Tree partitioning] \label{lem:tree decomp}
    Let $\TreeWeight > 0$ be a real value, and $\SetT$ be an out-tree satisfying $\NumEdges(\SetT) \geq \TreeWeight$. Then $\SetT$ can be partitioned into $\OptNumSub$ edge-disjoint out-subtrees $\SetT_{1}, \dots, \SetT_{\Delta}$ such that for each $\OptSubInd \in [1,\OptNumSub]\colon 1 \leq \NumEdges(\SetT_{\OptSubInd}) \leq \lfloor 2\TreeWeight \rfloor$, and $1 \leq \OptNumSub \leq \lfloor \frac{\NumEdges(\SetT)}{\TreeWeight} \rfloor$.
\end{lemma}

\section{\GR{}: Radius-Sensitive Approximation Algorithm} \label{sec:poly time approx algorithm}
We present the algorithm \GR{} for \DRCSMa{}, with pseudocode in Algorithm~\ref{alg:GR} and performance guarantees in Theorem~\ref{thm:GR}. We explain how to use \GR{} to achieve a feasible approximation for (Directed) \CSMa{} in Section~\ref{sec:feasible approx for DCSM} and a bicriteria approximation with $(1+\Errorv)$-violation factor for \DRCSMa{} in Section~\ref{sec:violation approx for DRCSM}.

\begin{algorithm}[ht]

\DontPrintSemicolon
\caption{\GR}
\label{alg:GR}

\KwIn{$\Graph = (\VertexSet, \EdgeSet)$: directed graph, $\Funcf : 2^\VertexSet \rightarrow \Reals_{\geq 0}$: value oracle, $\Card \geq 1$: size constraint, $\Vertv \in \VertexSet$: root, $\Height \in [1, \Card]$: radius of $\OptTree$, \Alg{}: $(\AppFactor, \VioFactor)$-approximation subroutine.}

\KwOut{$\SetS \subseteq \Graph$: an out-tree, with root $\Vertv$, that $( \frac{1}{2} \AppFactor(\Height), 4\VioFactor(\Height) )$-approximates $\OptTree$.}

$\SetS_{\leq 0} \gets (\{ \Vertv \}, \varnothing)$ \tcp*[r]{Initialize the output solution} \label{line:GR init sol}

$\SetW \gets $ set of vertices $\Vertw \in \VertexSet$ where $\Dist(\Vertv, \Vertw) \leq \Height - 1$\; \label{line:GR init dist r vertices}

$\NumSubtrees \gets \lfloor \frac{2 \Card}{\Height} \rfloor$\; \label{line:GR num updates}

\For(\tcp*[f]{Add approximate out-subtrees over $\NumSubtrees$ iterations}){$\IterInd = 1, \dots, \NumSubtrees$} {\label{line:GR update loop}
    \For{$\Vertw \in \SetW$}{\label{line:GR dist r vertex loop}
    
        $\SetS^{\Vertw}_{\IterInd} \gets \Alg( \Graph, \Funcf( \,\cdot \mid \SetS_{\leq \IterInd - 1}), \Height, \Vertw)$ \tcp*[r]{$\NumEdges(\SetS^{ \Vertw }_{\IterInd}) \leq \VioFactor(\Height) \Height$} \label{line:GR call RG}
    }
    
    $\Vertw_{\IterInd} \gets \argmax_{\Vertw \in \SetW} \Funcf(\SetS^{\Vertw}_{\IterInd} \mid \SetS_{\leq \IterInd - 1})$ \tcp*[r]{Select an out-subtree greedily} \label{line:GR select subtree greedily}

    $\SetP_{\IterInd} \gets$ shortest path from $\Vertv$ to $\Vertw_{\IterInd}$ \tcp*[r]{$\NumEdges(\SetP_{\IterInd}) \leq \Height \leq \VioFactor(\Height) \Height$} \label{line:GR path to w tree}
    
    $\SetS_{\leq \IterInd} \gets \SetS_{\leq \IterInd-1} \cup \SetS^{ \Vertw_{\IterInd} }_{\IterInd} \cup \SetP_{\IterInd}$ \tcp*[r]{$\NumEdges(\SetS^{ \Vertw_{\IterInd} }_{\IterInd} \cup \SetP_{\IterInd}) \leq 2\VioFactor(\Height) \Height$} \label{line:GR update sol}
}

$\SetS \gets \SetS_{\leq \IterInd}$ \;

\KwRet{$\SetS$}

\end{algorithm}

\subsection{Overview of \GR{}}
\GR{} takes as input a directed graph $\Graph$, a non-negative monotone submodular function $\Funcf$, a size constraint $\Card \geq 1$, a root vertex $\Vertv$, the optimal solution's radius $\Radius$, and an $(\AppFactor(\Card), \VioFactor(\Card))$-approximation subroutine \Alg{}; we assume that \Alg{} takes an instance $(\Graph, \Funcf, \Card, \Vertv)$ of \DRCSMa{} as input. \GR{} outputs an out-tree $\SetS \subseteq \Graph$, with root $\Vertv$, that $(\frac{1}{2} \AppFactor(\Height), 4\VioFactor(\Height) )$-approximates $\OptTree$.

To give a high-level overview of how \GR{} works, we first assume $\OptTree$ is partitioned into $\OptNumSub \leq \lfloor \frac{2 \Card}{\OptRadius} \rfloor$ edge-disjoint out-subtrees, $\OptTree_{1}, \dots, \OptTree_{\OptNumSub}$, each of size at most $\OptRadius$ (Corollary~\ref{cor:opt tree decomp}). The roots of these out-subtrees are all of distance $\OptRadius-1$ from $\Root$. \GR{} initializes the output solution, $\SetS$, with the root,~$\Root$, and updates it over $\NumSubtrees = \lfloor \frac{2 \Card}{\OptRadius} \rfloor$ iterations. Let $\SetS_{\leq 0} = (\{\Vertv\}, \varnothing)$ and, for each $\IterInd \in [1, \NumSubtrees]$, let $\SetS_{\leq \IterInd}$ be the $\IterInd$th partial solution, $\BestOptSub$ be the out-subtree from $\OptTree_{1}, \dots, \OptTree_{\OptNumSub}$ with maximum marginal gain to $\SetS_{\leq \IterInd-1}$, and $\Bestw$ be the root of $\BestOptSub$. Then, in the $\IterInd$th iteration, \GR{} performs these steps: (1) guess $\Vertw_{\IterInd} = \Bestw$, (2) call \Alg{} to construct an out-subtree $\SetS^{ \Vertw_{\IterInd} }_{\IterInd}$, with root $\Vertw_{\IterInd}$, that $(\AppFactor(\Radius), \VioFactor(\Radius))$-approximates $\BestOptSub$, and (3) add $\SetS^{ \Vertw_{\IterInd} }_{\IterInd}$ and a $\Vertv$--$\Vertw_{\IterInd}$ connecting path to $\SetS_{\leq \IterInd-1}$ to get $\SetS_{\leq \IterInd}$.

\subsection{Analysis of \GR{}}
Our main result is Theorem~\ref{thm:GR} below. The approximation factor follows from Lemma~\ref{lem:GR approx factor} and the violation factor follows from Lemma~\ref{lem:GR violation factor}. The running time easily follows from \GR{} making at most $\lfloor \frac{2\Card}{\OptRadius} \rfloor \NumVertices$ calls to both \Alg{} and $\Funcf$ (in the statement of Theorem~\ref{thm:GR}, $\Time(\NumVertices, \OptRadius)$ is the running time of \Alg{} with an input graph $\Graph$ of $\NumVertices$ vertices, and size constraint $\OptRadius$).

\begin{theorem} \label{thm:GR}
    Let $(\Graph, \Funcf, \Card, \Vertv)$ be an instance of \DRCSMa{} and \Alg{} be an $(\AppFactor(\Card), \VioFactor(\Card))$-approximation algorithm for \DRCSMa{} that runs in time $\Time(\NumVertices, \Card)$. Then \GR{}, with subroutine \Alg{}, is a $( \frac{1}{2} \AppFactor(\Height), 4\VioFactor(\Height) )$-approximation algorithm for \DRCSMa{} that runs in time $O(\frac{\Card\NumVertices}{\OptRadius} \Time(\NumVertices, \Radius))$.
\end{theorem}

\paragraph*{Analysis Notation for \GR{}.}
We use the following notation based on the pseudocode of \GR{} (Algorithm~\ref{alg:GR}). Recall that $\OptHeight$ denotes the radius of the optimal out-tree $\OptTree$. Let $\NumSubtrees = \lfloor \frac{2 \Card}{\Height} \rfloor$ denote the total number of updates \GR{} makes in Line~\ref{line:GR update sol} to construct the output solution. Let $\SetS_{\leq 0} = ( \{ \Vertv \}, \varnothing)$ denote the initial state of the output solution as in Line~\ref{line:GR init sol}, and $\SetS_{\leq \IterInd}$ denote the $\IterInd$th partial solution as in Line~\ref{line:GR update sol} of the $\IterInd$th iteration of the Line~\ref{line:GR update loop} loop. Further, let $\SetS^{\Vertw_{\IterInd}}_{\IterInd}$ denote the $\IterInd$th approximate out-subtree appended to the output solution, and $\SetP_{\IterInd}$ denote the path added to connect $\Vertv$ to $\Vertw_{\IterInd}$, the root of $\SetS^{\Vertw_{\IterInd}}_{\IterInd}$.

\paragraph*{Partitioning the Optimal Solution for \GR{}.} 
First, an optimal out-tree $\OptTree$ can be partitioned into edge-disjoint out-subtrees each of size at most $\Height$ as in Corollary~\ref{cor:opt tree decomp} below. This corollary directly follows from Lemma~\ref{lem:tree decomp} by setting $\TreeWeight = \frac{\Height}{2}$ and $\SetT = \OptTree$, where $\NumEdges(\OptTree) = \Card$.

\begin{corollary} \label{cor:opt tree decomp}
    Let $\Height \geq 1$ be an integer. Then $\OptTree$ can be partitioned into $\OptNumSub$ edge-disjoint out-subtrees $\OptTree_{1}, \dots, \OptTree_{\Delta}$ such that for each $\OptSubInd \in [1,\OptNumSub] \colon 1 \leq \NumEdges(\OptTree_{\OptSubInd}) \leq \Height$, and $1 \leq \OptNumSub \leq \lfloor \frac{2 \Card}{\Height} \rfloor$.
\end{corollary}

\paragraph*{Approximation Factor of \GR{}.}
\begin{lemma} \label{lem:GR approx factor}
    Let $(\Graph, \Funcf, \Card, \Vertv)$ be an instance of \DRCSMa{}. Then \GR{} outputs an out-tree $\SetS$ satisfying $\Funcf(\SetS) \geq \frac{1}{2} \AppFactor(\Height) \Funcf(\OptTree)$.
\end{lemma}
\begin{proof}
    As shown by Corollary~\ref{cor:opt tree decomp}, $\OptTree$ can be partitioned into $\OptNumSub$ edge-disjoint out-subtrees $\OptTree_{1}, \dots, \OptTree_{\OptNumSub}$ each of size at most $\Height$, where $1 \leq \OptNumSub \leq \lfloor \frac{2 \Card}{\Height} \rfloor$. Assume without loss of generality that $\OptNumSub = \NumSubtrees = \lfloor \frac{2 \Card}{\Height} \rfloor$, recalling that $\NumSubtrees$ denotes the total number of updates in Line~\ref{line:GR update sol} to construct $\SetS$.
    
    Now consider the $\IterInd$th iteration of the Line~\ref{line:GR update loop} update loop. Let $\BestOptSub$ be the subtree amongst $\OptTree_{1}, \dots, \OptTree_{\NumSubtrees}$ that maximizes $\Funcf(\BestOptSub \mid \SetS_{\leq \IterInd - 1})$. By a standard analysis from the submodularity of $\Funcf$ and the maximality of $\BestOptSub$, Ineq.~\eqref{eqn:opt subtree gain} below holds.
    \begin{align}
        \Funcf( \BestOptSub \mid \SetS_{\leq \IterInd - 1}) \geq \frac{1}{\NumSubtrees} ( \Funcf( \OptTree ) - \Funcf( \SetS_{\leq \IterInd - 1} ) ). \label{eqn:opt subtree gain}
    \end{align}
    
    Let $\Bestw$ be the root of $\BestOptSub$. It holds that $\Bestw$ is within distance $\Radius-1$ of $\Vertv$, so \GR{} must guess $\Vertw = \Bestw$ in the inner Line~\ref{line:GR dist r vertex loop} loop. Also, $\BestOptSub$ has size at most $\Height$ by definition. Thus, $\BestOptSub$ is a feasible solution to the instance $(\Graph, \Funcf( \,\cdot \mid \SetS_{\leq \IterInd - 1}), \Height, \Vertw)$ of \DRCSMa{}. This means that the corresponding call to \Alg{} in Line~\ref{line:GR call RG} outputs an out-tree $\SetS^{\Bestw}_{\IterInd}$ satisfying
    \begin{align}
        \Funcf( \SetS^{\Bestw}_{\IterInd} \mid \SetS_{\leq \IterInd - 1}) \geq \AppFactor(\Height) \Funcf( \BestOptSub \mid \SetS_{\leq \IterInd - 1}). \label{eqn:subtree approx factor}
    \end{align}
    In the $\IterInd$th Line~\ref{line:GR update sol} update, \RA{} actually adds the out-subtree with maximum marginal gain along with its connecting path, namely $\SetS^{\Vertw_{\IterInd}}_{\IterInd} \cup \SetP_{\IterInd}$. Hence, we prove Ineq.~\eqref{eqn:partial sol gap} for all iterations $\IterInd \in [1, \NumSubtrees]$ below; the 1st inequality holds by the monotonicity of $\Funcf$, the 2nd by the maximality of $\SetS^{\Vertw_{\IterInd}}_{\IterInd}$, the 3rd by Ineq.~\eqref{eqn:subtree approx factor}, and the 4th by Ineq.~\eqref{eqn:opt subtree gain}.
    \begin{align}
        \Funcf(\SetS^{\Vertw_{\IterInd}}_{\IterInd} \cup \SetP_{\IterInd} \mid \SetS_{\leq \IterInd - 1}) &\geq \Funcf(\SetS^{\Vertw_{\IterInd}}_{\IterInd} \mid \SetS_{\leq \IterInd - 1})
        \geq \Funcf(\SetS^{\Bestw}_{\IterInd} \mid \SetS_{\leq \IterInd - 1}) \notag \\
        &\geq \AppFactor(\Height) \Funcf( \BestOptSub \mid \SetS_{\leq \IterInd - 1}) \notag \\
        &\geq \frac{\AppFactor(\Height)}{\NumSubtrees} ( \Funcf( \OptTree ) - \Funcf( \SetS_{\leq \IterInd - 1} ) ), \notag \\
        \Funcf(\OptTree) - \Funcf(\SetS_{\leq \IterInd}) &\leq \left( 1 - \frac{\AppFactor(\Height)}{\NumSubtrees} \right) (  \Funcf(\OptTree) - \Funcf(\SetS_{\leq \IterInd-1}) ). \label{eqn:partial sol gap}
    \end{align}

    Finally, \GR{} makes $\NumSubtrees$ updates to construct $\SetS$, so we can chain Ineq.~\eqref{eqn:partial sol gap} $\NumSubtrees$ times. Thus, we prove the lemma below; the 2nd inequality holds by the non-negativity of $\Funcf$. 
    \begin{align}
        \Funcf(\OptTree) - \Funcf(\SetS) &\leq \left( 1 - \frac{\AppFactor(\Height)}{\NumSubtrees} \right)^{\NumSubtrees} \left(  \Funcf(\OptTree) - \Funcf(\SetS_{\leq 0}) \right) \notag \\
        &\leq \left( 1 - \frac{\AppFactor(\Height)}{\NumSubtrees} \right)^{\NumSubtrees} \Funcf(\OptTree) \notag \\
        &\leq e^{-\AppFactor(\Height)} \Funcf(\OptTree) \notag
        \leq \left(1- \frac{ \AppFactor(\Height) }{2} \right) \Funcf(\OptTree), \notag \\
        \Funcf(\SetS) &\geq \frac{ \AppFactor(\Height) }{2} \Funcf(\OptTree). \tag*{\qedhere}
    \end{align}    
\end{proof}

\paragraph*{Violation Factor of \GR{}.}
\begin{lemma} \label{lem:GR violation factor}
    Let $(\Graph, \Funcf, \Card, \Vertv)$ be an instance of \DRCSMa{}. Then \GR{} outputs an out-tree $\SetS$ satisfying $\NumEdges(\SetS) \leq 4\VioFactor(\Height) \Card$.
\end{lemma}
\begin{proof}
    For each $\IterInd \in [1, \NumSubtrees]$, the out-subtree $\SetS_{\IterInd}^{\Vertw_{\IterInd}}$ is constructed by the subroutine \Alg{} with size constraint $\Height$ (Line~\ref{line:GR call RG} of \GR{}). Thus, it holds that $\NumEdges( \SetS_{\IterInd}^{\Vertw_{\IterInd}} ) \leq \VioFactor(\Height) \Height$. Further, the connecting path $\SetP_{\IterInd}$ has size $\NumEdges(\SetP_{\IterInd}) \leq \Height \leq \VioFactor(\Height) \Height$. Therefore, $\NumEdges( \SetS_{\IterInd}^{\Vertw_{\IterInd}} \cup \SetP_{\IterInd}) \leq 2 \VioFactor(\Height) \Height$. From this, we bound the size of $\SetS$ below, recalling that $\NumSubtrees = \lfloor \frac{2\Card}{\Height} \rfloor$.
    \begin{align}
        \NumEdges(\SetS) \leq \sum_{\IterInd = 1}^{\NumSubtrees} \NumEdges( \SetS_{\IterInd}^{\Vertw_{\IterInd}} \cup \SetP_{\IterInd})
        \leq 2 \NumSubtrees \VioFactor(\Height) \Height \leq \frac{4 \Card}{\Height} \VioFactor(\Height) \Height = 4 \VioFactor(\Height) \Card. \tag*{\qedhere}
    \end{align}
\end{proof}

\subsection{Feasible Approximation for (Directed) \CSMa{}} \label{sec:feasible approx for DCSM}
Here we explain how \GR{} can be used to achieve a feasible $\frac{\AppFactor(\Height)}{16\VioFactor(\Height)}$-approximation for the problem of (Directed) \CSMa{}.

First, an out-tree $\SetS$ can be partitioned into out-subtrees as in Corollary~\ref{cor:sol tree decomp unrooted} below. This corollary directly follows from Lemma~\ref{lem:tree decomp} by setting $\TreeWeight = \frac{\Card}{2}$ and $\SetT = \SetS$, where $\NumEdges(\SetS) \leq 4 \VioFactor(\Height) \Card$.

\begin{corollary} \label{cor:sol tree decomp unrooted}
    Let $\Height \geq 1$ be an integer. Then $\SetS$ can be partitioned into $\OptNumSub$ edge-disjoint out-subtrees $\SetS_{1}, \dots, \SetS_{\Delta}$ such that for each $\OptSubInd \in [1,\OptNumSub] \colon 1 \leq \NumEdges(\SetS_{\OptSubInd}) \leq \Card$, and $1 \leq \OptNumSub \leq \lfloor 8 \VioFactor(\Height) \rfloor$.
\end{corollary}

Now we prove the required approximation result in Corollary~\ref{cor:GR for DCSM}, which follows from Theorem~\ref{thm:GR} and Corollary~\ref{cor:sol tree decomp unrooted}.

\begin{corollary}\label{cor:GR for DCSM}
     Let $(\Graph, \Funcf, \Card)$ be an instance of (Directed) \CSMa{} and \Alg{} be an $(\AppFactor(\Card), \VioFactor(\Card))$-approximation algorithm for \DRCSMa{} that runs in time $\Time(\NumVertices, \Card)$. Then there exists an $\frac{\AppFactor(\Height)}{16\VioFactor(\Height)}$-approximation algorithm for (Directed) \CSMa{} that runs in time $O(\frac{\Card\NumVertices}{\OptRadius} \Time(\NumVertices, \Height))$.
\end{corollary}
\begin{proof}
    Given an instance $(\Graph, \Funcf, \Card)$ of \CSMa{} (or \DCSMa{}) and its optimal out-tree $\OptTree$, guess $\Root$ and $\Radius$ as the center (or root) vertex and the radius of $\OptTree$, respectively.
    
    First run \GR{} with input $(\Graph, \Funcf, \Card, \Root, \Height, \Alg)$. By Theorem~\ref{thm:GR}, this gives an out-tree $\SetS$ with value $\Funcf(\SetS) \geq \frac{1}{2} \AppFactor(\Height) \Funcf(\OptTree)$ and size $\NumEdges(\SetS) \leq 4 \VioFactor(\Height) \Card$ in time $O(\frac{\Card\NumVertices}{\OptRadius} \Time(\NumVertices, \Height))$.
    
    Now partition $\SetS$ into $\OptNumSub \leq \lfloor 8 \VioFactor(\Height) \rfloor$ edge-disjoint out-subtrees $\SetS_{1}, \dots, \SetS_{\Delta}$ as in Corollary~\ref{cor:sol tree decomp unrooted}. Let $\BestSolSub$ be the out-subtree with maximum value $\Funcf(\BestSolSub)$, which is the final out-tree as required. We have that $\BestSolSub$ is a feasible out-tree as $\NumEdges(\BestSolSub) \leq \Card$. Further, by the submodularity and non-negativity of $\Funcf$, and the bound of $\OptNumSub \leq \lfloor 8 \VioFactor(\Height) \rfloor$, $\BestSolSub$ has the required approximation factor as shown below.
    \begin{align}
         \Funcf(\BestSolSub) \geq \frac{1}{\OptNumSub} \Funcf(\SetS) \geq \frac{1}{8 \VioFactor(\Height)} \Funcf(\SetS) \geq \frac{\AppFactor(\Height)}{16 \VioFactor(\Height)} \Funcf(\OptTree). \tag*{\qedhere}
    \end{align} 
\end{proof}

\subsection{Bicriteria Approximation for \DRCSMa{} with \texorpdfstring{$(1+\Errorv)$}{(1 + \textit{delta})}-Violation Factor} \label{sec:violation approx for DRCSM}
Here we explain how \GR{} can be used to achieve, for every $\Errorv \in [\frac{1}{\Card}, 1]$, a $(  \frac{\Errorv \AppFactor(\Height)}{16 \VioFactor(\Height)}, 1 + \Errorv)$-approximation for \DRCSMa{}.

First, an out-tree $\SetS$ can be partitioned into out-subtrees as in Corollary~\ref{cor:sol tree decomp rooted} below. This corollary directly follows from Lemma~\ref{lem:tree decomp} by setting $\TreeWeight = \frac{\Errorv \Card}{2}$ and $\SetT = \SetS$, where $\NumEdges(\SetS) \leq 4 \VioFactor(\Height) \Card$. 

\begin{corollary} \label{cor:sol tree decomp rooted}
    Let $\Height \geq 1$ be an integer and $\Errorv \in [\frac{1}{\Card}, 1]$ be a value. Then $\SetS$ can be partitioned into $\OptNumSub$ edge-disjoint out-subtrees $\SetS_{1}, \dots, \SetS_{\Delta}$ such that for each $\OptSubInd \in [1,\OptNumSub] \colon 1 \leq \NumEdges(\SetS_{\OptSubInd}) \leq \Errorv \Card$, and $1 \leq \OptNumSub \leq \lfloor \frac{8 \VioFactor(\Height)}{\Errorv} \rfloor$.
\end{corollary}

Now we prove the required approximation result in Corollary~\ref{cor:GR for DRCSM}, which follows from Theorem~\ref{thm:GR} and Corollary~\ref{cor:sol tree decomp rooted}.

\begin{corollary}\label{cor:GR for DRCSM}
    Let $(\Graph, \Funcf, \Card, \Vertv)$ be an instance of \DRCSMa{} and \Alg{} be an $(\AppFactor(\Card), \VioFactor(\Card))$-approximation algorithm for \DRCSMa{} that runs in time $\Time(\NumVertices, \Card)$. Then, for every $\Errorv \in [\frac{1}{\Card}, 1]$, there exists a $( \frac{\Errorv \AppFactor(\Height)}{16 \VioFactor(\Height)}, 1 + \Errorv)$-approximation algorithm for \DRCSMa{} that runs in time $O(\frac{\Card\NumVertices}{\OptRadius} \Time(\NumVertices, \Height))$.
\end{corollary}
\begin{proof}
    Given an instance $(\Graph, \Funcf, \Card, \Vertv)$ of \DRCSMa{}, assume that we prune the input graph $\Graph$ so that it only contains those vertices $\Vertw$ with $\Dist(\Vertv, \Vertw) \leq \Card$ (which preserves the optimal solution), and guess  $\Radius$ as the radius of $\OptTree$.

    First run \GR{} with input $(\Graph, \Funcf, \Card, \Root, \Height, \Alg)$. By Theorem~\ref{thm:GR}, this gives an out-tree $\SetS$ with value $\Funcf(\SetS) \geq \frac{1}{2} \AppFactor(\Height) \Funcf(\OptTree)$ and size $\NumEdges(\SetS) \leq 4 \VioFactor(\Height) \Card$ in time $O(\frac{\Card\NumVertices}{\OptRadius} \Time(\NumVertices, \Height))$.
    
    Now partition $\SetS$ into $\OptNumSub$ edge-disjoint out-subtrees $\SetS_{1}, \dots, \SetS_{\Delta}$ as in Corollary~\ref{cor:sol tree decomp rooted}. Let $\SetS_{\BestOptSubInd}$ be the out-subtree with maximum value, and $\Vertw_{\BestOptSubInd}$ be its root. Let $\BestSolSub$ be the new out-subtree formed by combining a shortest $\Vertv$--$\Vertw_{\BestOptSubInd}$ path with $\SetS_{\BestOptSubInd}$, giving the final out-tree as required. By the monotonicity, submodularity, and non-negativity of $\Funcf$, and the bound of $\OptNumSub \leq \lfloor \frac{8 \VioFactor(\Height)}{\Errorv} \rfloor$, $\BestSolSub$ has the required approximation factor as shown below.
    \begin{align}
         \Funcf(\BestSolSub) \geq \Funcf(\SetS_{\BestOptSubInd}) \geq \frac{1}{\OptNumSub} \Funcf(\SetS) \geq \frac{\Errorv}{8 \VioFactor(\Height)} \Funcf(\SetS) \geq \frac{\Errorv \AppFactor(\Height)}{16 \VioFactor(\Height)} \Funcf(\OptTree). \notag
    \end{align} 
    Further, $\BestSolSub$ has the required violation factor since the shortest $\Vertv$--$\Vertw_{\BestOptSubInd}$ path has at most $\Card$ edges (by the pruning of $\Graph$), and $\NumEdges(\SetS_{\BestOptSubInd}) \leq \Errorv \Card$, giving 
    \begin{align}
        \NumEdges(\BestSolSub) = \Dist(\Vertv, \Vertw_{\BestOptSubInd}) + \NumEdges(\SetS_{\BestOptSubInd})\leq \Card + \Errorv \Card = (1+\Errorv) \Card. \tag*{\qedhere}
    \end{align}
\end{proof}

\section{\texorpdfstring{\RAd{}}{\RA-\textit{d}}: Recursive Greedy Approximation Algorithm} \label{sec:rec greedy approx algorithm}
We present the algorithm \RA{} for \DRCSMa{}, with pseudocode in Algorithm~\ref{alg:RA} and performance guarantees in Theorem~\ref{thm:RA}. We also define \RAd{} to take an instance $(\Graph, \Funcf, \Card, \Vertv)$ of \DRCSMa{} and output the solution from \RA$(\Graph, \Funcf, \Card, \Vertv, \Card^{\frac{1}{\RecDepth}})$. We give the performance guarantees of \RAd{} in Theorem~\ref{thm:RAd}.

\begin{algorithm}
\DontPrintSemicolon
\caption{\RA}
\label{alg:RA}

\KwIn{$\Graph=(\VertexSet, \EdgeSet)$: directed graph, $\Funcf : 2^\VertexSet \rightarrow \Reals_{\geq 0}$: value oracle, $\TargetSize \geq 1$: size constraint, $\Vertv \in \VertexSet$: root, $\TargetDiv > 1$: size divisor.}

\KwOut{$\SetS \subseteq \Graph$: an out-tree, with root $\Vertv$, that $( \frac{1}{\RecLevel+1}, 1 + 3\RecLevel\TargetDiv + \RecLevel \log_{3/2} ( \frac{\TargetSize}{\TargetDiv} ) )$-approximates $\OptTree$, where $\RecLevel$ is the integer satisfying $\TargetDiv^{\RecLevel-1} < \TargetSize \leq \TargetDiv^{\RecLevel}$.}

\If(\tcp*[f]{Base case: greedily add an out-neighbor of the root}){$\TargetSize = 1$}{ \label{line:RG base case}
    $\SetW \gets$ set of out-neighbors of $\Vertv$ \; \label{line:RG neighbors of v}
    $\Bestwa \gets \argmax_{\Vertw \in \SetW} \Funcf( \{\Vertv, \Vertw\} )$ \; \label{line:RG base case select w}
    $\SetS \gets \left( \{ \Vertv, \Bestwa \}, \{ (\Vertv, \Bestwa) \} \right)$ \; \label{line:RG base case update sol}
}

\Else(\tcp*[f]{Recursive case: add approximate out-subtrees until all $\TargetSize$ edges of $\OptTree$ have been approximated}){ \label{line:RG inductive case}

$\SetS_{\leq 0} \gets (\{ \Vertv \}, \varnothing)$ \tcp*[r]{Initialize the output out-tree} \label{line:RG init sol}

$\SetW \gets$ set of vertices $\Vertw \in \VertexSet$ where $\Dist(\Vertv, \Vertw) \leq \TargetSize - 1$ \; \label{line:RG set of rec centers}

$\TargetSize_{0} \gets \TargetSize$ \;

$\IterInd \gets 0$

\While{$\TargetSize_{\IterInd} > 0$}{ \label{line:RG update loop}

    $\IterInd \gets \IterInd + 1$

    $\TargetSubsizeMin \gets \lceil \frac{1}{3} \max\{ \min\{ \frac{\TargetSize}{\TargetDiv}, \TargetSize_{\IterInd-1}\}, 1 \} \rceil $\; \label{line:RG assign target subsize min}

    $\TargetSubsizeMax \gets \lfloor \max\{ \min\{ \frac{\TargetSize}{\TargetDiv}, \TargetSize_{\IterInd-1}\}, 1 \} \rfloor $ \; \label{line:RG assign target subsize max}
    
    \For{$\Vertw \in \SetW$ \textnormal{and} $\TargetSubsize = \TargetSubsizeMin, \dots, \TargetSubsizeMax $}{ \label{line:RG guess w and c}
        
        $\SetS^{\Vertw, \TargetSubsize}_{\IterInd} \gets \RA(\Graph, \Funcf( \, \cdot \mid \SetS_{\leq \IterInd-1}), \TargetSubsize, \Vertw, \TargetDiv)$ \; \label{line:RG find w tree}   
    }
    $\Vertw_{\IterInd}, \TargetSubsize_{\IterInd} \gets \argmax_{ \Vertw \in \SetW, \TargetSubsize \in [ \TargetSubsizeMin, \TargetSubsizeMax ] } \frac{1}{\TargetSubsize} \Funcf(\SetS^{\Vertw, \TargetSubsize}_{\IterInd} \mid \SetS_{\leq \IterInd-1}) $ \tcp*[r]{Select an out-subtree greedily by marginal density} \label{line:RG select subtree greedily}

    $\SetP_{\IterInd} \gets $ shortest path from $\Vertv$ to $\Vertw_{\IterInd}$ \tcp*[r]{$\NumEdges( \SetP_{\IterInd} ) \leq \TargetSize$} \label{line:RG v to w path}

    $\SetS_{\leq \IterInd} \gets \SetS_{\leq \IterInd-1} \cup \SetS^{\Vertw_{\IterInd}, \TargetSubsize_{\IterInd}}_{\IterInd} \cup \SetP_{\IterInd}$ \label{line:RG update sol} \;
    
    $\TargetSize_{\IterInd} \gets \TargetSize_{\IterInd-1} - \TargetSubsize_{\IterInd}$ \label{line:RG update target size rem} \;
}
$\SetS \gets \SetS_{\leq \IterInd}$\;
}

\KwRet{$\SetS$}
\end{algorithm}

\subsection{Overview of \RA{}}
\RA{} takes as input a directed graph $\Graph$, a non-negative monotone submodular function $\Funcf$, a size constraint $\TargetSize \geq 1$, a root vertex $\Vertv$, and a size divisor $\TargetDiv > 1$. We define the \emph{recursion level} of a call to \RA{} to be the integer $\RecLevel$ satisfying $\TargetDiv^{\RecLevel-1} < \TargetSize \leq \TargetDiv^{\RecLevel}$. \RA{} outputs an out-tree $\SetS \subseteq \Graph$, with root $\Vertv$, that $( \frac{1}{\RecLevel+1}, 1 + 3\RecLevel\TargetDiv + \RecLevel \log_{3/2} ( \frac{\TargetSize}{\TargetDiv} ) )$-approximates $\OptTree$.

In the base case where $\TargetSize = 1$, \RA{} simply returns an out-tree $\SetS$ consisting of the root, $\Root$, and an out-neighbor, $\Bestwa$, with maximum value.

To give a high-level overview of how \RA{} works in the recursive case where $\TargetSize \geq 2$, we first assume $\OptTree$ can be partitioned into $\OptNumSub(\IterInd)$ out-subtrees whose sizes depend on the $\IterInd$th iteration of the Line~\ref{line:RG update loop} loop. That is, $\OptTree$ can be partitioned into edge-disjoint out-subtrees, $\OptTree_{\IterInd, 1}, \dots, \OptTree_{\IterInd, \OptNumSub(\IterInd)}$, each of size in the range $[\TargetSubsizeMin, \TargetSubsizeMax]$, where the values of $\TargetSubsizeMin$ and $\TargetSubsizeMax$ depend on the $\IterInd$th iteration.
Note that the roots of the out-subtrees $\OptTree_{\IterInd, 1}, \dots, \OptTree_{\IterInd, \OptNumSub(\IterInd)}$ are all of distance $\TargetSize-1$ from $\Root$.

\RA{} initializes the output solution, $\SetS$, with the root,~$\Root$, and updates it while the remaining number of edges to approximate, $\TargetSize_{\IterInd}$, is non-zero. Let $\SetS_{\leq 0} = (\{\Vertv\}, \varnothing)$ and, for each iteration $\IterInd$, let $\SetS_{\leq \IterInd}$ be the $\IterInd$th partial solution, $\BestOptSub$ be the out-subtree from $\OptTree_{\IterInd, 1}, \dots, \OptTree_{\IterInd, \OptNumSub(\IterInd)}$ with maximum marginal density to $\SetS_{\leq \IterInd-1}$, $\Bestw$ be the root of $\BestOptSub$, and $\BestSubsize$ be the size of $\BestOptSub$. Then, in the $\IterInd$th iteration, \RA{} performs these steps: (1) guess $\Vertw_{\IterInd} = \Bestwa$ and $\TargetSubsize_{\IterInd} = \BestSubsize$, (2) make a recursive call to construct an out-subtree $\SetS^{\Vertw_{\IterInd}, \TargetSubsize_{\IterInd}}_{\IterInd}$, with root $\Vertw_{\IterInd}$, that is a bicriteria approximation of $\BestOptSub$, and (3) add $\SetS^{\Vertw_{\IterInd}, \TargetSubsize_{\IterInd}}_{\IterInd}$ and a $\Vertv$--$\Vertw_{\IterInd}$ connecting path to $\SetS_{\leq \IterInd-1}$ to get $\SetS_{\leq \IterInd}$.

\subsection{Analysis of \RA{}} \label{sec:RA analysis}
We state the performance guarantees of \RA{} in Theorem~\ref{thm:RA} below. We give a tree partitioning in Lemma~\ref{lem:balanced tree decomp} and Corollary~\ref{cor:target tree decomp} in order to prove the approximation factor in Lemma~\ref{lem:RG approx of sol}. The violation factor follows from Lemma~\ref{lem:RG sol size}, and the running time follows from Lemma~\ref{lem:RG running time}.

\begin{theorem} \label{thm:RA}
    Let $(\Graph, \Funcf, \TargetSize, \Vertv)$ be an instance of \DRCSMa{}. Let $\TargetDiv > 1$ be the size divisor given to \RA{} and $\RecLevel$ be the integer satisfying $\TargetDiv^{\RecLevel-1} < \TargetSize \leq \TargetDiv^{\RecLevel}$. Then \RA{} is a $( \frac{1}{\RecLevel+1}, 1 + 3\RecLevel\TargetDiv + \RecLevel \log_{3/2} ( \frac{\TargetSize}{\TargetDiv} ) )$-approximation algorithm for \DRCSMa{} and runs in time $O(\NumVertices^{\RecLevel+1} \TargetSize^{2\RecLevel+2})$.
\end{theorem}

\paragraph*{Analysis Notation for \RA{}.}
We use the following notation throughout the analysis of \RA{}, particularly in the recursive case. The \emph{recursion level} (of a given \RA{} call) is the integer $\RecLevel$ satisfying $\TargetDiv^{\RecLevel-1} < \TargetSize \leq \TargetDiv^{\RecLevel}$. Let $\NumSubtrees$ denote the total number of updates \RA{} makes in Line~\ref{line:RG update sol} to construct the output solution. Let $\SetS_{\leq 0} = ( \{ \Vertv \}, \varnothing)$ denote the initial state of the output solution as in Line~\ref{line:RG init sol}, and $\SetS_{\leq \IterInd}$ denote the $\IterInd$th partial solution as in Line~\ref{line:RG update sol} of the $\IterInd$th iteration of the Line~\ref{line:RG update loop} loop. Further, let $\SetS^{\Vertw_{\IterInd}, \TargetSubsize_{\IterInd}}_{\IterInd}$ denote the $\IterInd$th approximate out-subtree added to the output solution, and $\SetP_{\IterInd}$ denote the path appended to connect $\Vertv$ to $\Vertw_{\IterInd}$, i.e., the root of $\SetS^{\Vertw_{\IterInd}, \TargetSubsize_{\IterInd}}_{\IterInd}$.

\paragraph*{Partitioning the Optimal Solution for \RA{}.}
First, for each iteration $\IterInd$ of the Line~\ref{line:RG update loop} loop, we show that $\OptTree$ can be partitioned into edge-disjoint out-subtrees each of size in the range $[\TargetSubsizeMin, \TargetSubsizeMax]$, where $\TargetSubsizeMin = \lceil \frac{1}{3} \max\{ \min\{ \frac{\TargetSize}{\TargetDiv}, \TargetSize_{\IterInd-1} \}, 1 \} \rceil$ and $\TargetSubsizeMax = \lfloor \frac{1}{3} \max\{ \min\{ \frac{\TargetSize}{\TargetDiv}, \TargetSize_{\IterInd-1} \}, 1 \} \rfloor$ as in Lines~\ref{line:RG assign target subsize min} and \ref{line:RG assign target subsize max}. We use the following tree partitioning lemma based on balanced separators for trees.

\begin{lemma}[Tree partitioning with balanced separators] \label{lem:balanced tree decomp}
    Let $\TreeWeight > 0$ be a value, and $\SetT$ be an out-tree satisfying $\NumEdges(\SetT) \geq \TreeWeight$. Then $\SetT$ can be partitioned into $\OptNumSub$ edge-disjoint out-subtrees $\SetT_{1}, \dots, \SetT_{\Delta}$ such that for each $\OptSubInd \in [1,\OptNumSub] \colon \TreeWeight \leq \NumEdges(\SetT_{\OptSubInd}) \leq 3 \TreeWeight$.
\end{lemma}
\begin{proof}
    Every out-tree $\SetTp$ with $\NumEdges(\SetTp) \geq 2$ has a $\frac{1}{3}$-$\frac{2}{3}$ separator \cite{Lipton1979}, i.e., a vertex $\Vertv$ that splits $\SetTp$ into two edge-disjoint out-subtrees $\SetTp_1$ and $\SetTp_2$ such that $\frac{1}{3}\NumEdges(\SetTp) \leq \NumEdges(\SetTp_1) \leq \NumEdges(\SetTp_2) \leq \frac{2}{3}\NumEdges(\SetTp)$. Therefore, we can partition $\SetT$ using the following procedure.
    \begin{enumerate}
        \item Initialize a list of out-subtrees, $\List$, containing only $\SetT$.
        \item Repeat the following until every $\SetTp \in \List$ satisfies $\TreeWeight \leq \NumEdges(\SetTp) \leq 3\TreeWeight$: for each $\SetTp \in \List$ satisfying $\NumEdges(\SetTp) > 3\TreeWeight$, split it into $\SetTp_1$ and $\SetTp_2$ by finding its $\frac{1}{3}$-$\frac{2}{3}$ separator, and replace it with $\SetTp_1$ and $\SetTp_2$.
    \end{enumerate}

    The final list $\List$ gives $\SetT_{1}, \dots, \SetT_{\Delta}$ as required.
\end{proof}

We now prove Corollary~\ref{cor:target tree decomp} for partitioning $\OptTree$.

\begin{corollary} \label{cor:target tree decomp}
    Let $\TargetSize \geq \TargetSize_{\IterInd-1} \geq 1$ be integers, $\TargetDiv > 1$ be a real value, and $\OptTree$ be a tree satisfying $\NumEdges(\OptTree) = \TargetSize$. Then $\OptTree$ can be partitioned into $\OptNumSub$ edge-disjoint out-subtrees $\OptTree_{1}, \dots, \OptTree_{\OptNumSub}$ such that for each $\OptSubInd \in [1,\OptNumSub] \colon \lceil \frac{1}{3} \max\{ \min\{ \frac{\TargetSize}{\TargetDiv}, \TargetSize_{\IterInd-1} \}, 1 \} \rceil \leq \NumEdges(\OptTree_{\OptSubInd}) \leq \lfloor \max\{ \min\{\frac{\TargetSize}{\TargetDiv}, \TargetSize_{\IterInd-1} \}, 1 \} \rfloor$.
\end{corollary}
\begin{proof}
    In the case where $\TargetSize \leq \TargetDiv$ or $\TargetSize_{\IterInd-1} \leq 1$, simply make each edge in $\OptTree$ an out-subtree. Then it holds that for each $\OptSubInd \in [1,\OptNumSub] \colon \NumEdges(\OptTree_{\OptSubInd}) = 1 = \lfloor \max\{ \min\{ \frac{\TargetSize}{\TargetDiv}, \TargetSize_{\IterInd-1} \}, 1 \} \rfloor$.
    
    In the case where $\TargetSize > \TargetDiv$ and $\TargetSize_{\IterInd-1} > 1$, let $\TreeWeight = \frac{1}{3} \min\{ \frac{\TargetSize}{\TargetDiv}, \TargetSize_{\IterInd-1} \} $, satisfying both $\TreeWeight > 0$ and $\NumEdges(\SetT) = \TargetSize \geq \TreeWeight$. It follows from Lemma~\ref{lem:balanced tree decomp} that $\OptTree$ can be partitioned into $\OptNumSub$ edge-disjoint out-subtrees $\OptTree_{1}, \dots, \OptTree_{\OptNumSub}$ such that for each $\OptSubInd \in [1,\OptNumSub]\colon \frac{1}{3} \min\{ \frac{\TargetSize}{\TargetDiv}, \TargetSize_{\IterInd-1} \} \leq \NumEdges(\OptTree_{\OptSubInd}) \leq \min\{ \frac{\TargetSize}{\TargetDiv}, \TargetSize_{\IterInd-1} \}$. Since $\NumEdges(\OptTree_{\OptSubInd})$ is an integer at least 1, these bounds simplify to $\lceil \frac{1}{3} \max\{ \min\{ \frac{\TargetSize}{\TargetDiv}, \TargetSize_{\IterInd-1} \}, 1 \} \rceil \leq \NumEdges(\OptTree_{\OptSubInd}) \leq \lfloor \max\{ \min\{ \frac{\TargetSize}{\TargetDiv}, \TargetSize_{\IterInd-1} \}, 1 \} \rfloor$. This proves the corollary in this case, concluding the proof.
\end{proof}

\paragraph*{Approximation Factor of \RA{}.}
\begin{lemma}
\label{lem:RG approx of sol}
    Let $(\Graph, \Funcf, \TargetSize, \Vertv)$ be an instance of \DRCSMa{}. For a given call to \RA{}, let $\RecLevel$ be the integer satisfying $\TargetDiv^{\RecLevel-1} < \TargetSize \leq \TargetDiv^{\RecLevel}$. Then \RA{} outputs an out-tree $\SetS$ satisfying
    \begin{align*} 
        \Funcf(\SetS) \geq \frac{1}{\RecLevel+1}\Funcf(\OptTree).
    \end{align*}
\end{lemma}
\begin{proof}
    We prove the lemma by induction on the recursion level $\RecLevel$.

    \paragraph*{Base case \texorpdfstring{($\RecLevel=0$)}{(\textit{l} = 0)}.}
    In this case, $\SetS$ is constructed in the Line~\ref{line:RG base case} block by connecting the root $\Vertv$ and a neighbor $\Bestwa$ maximizing $\Funcf( \{ \Vertv, \Bestwa \} )$. Further, $\OptTree$ has root $\Vertv$ and size $\TargetSize \leq \TargetDiv^{0} = 1$. Hence, $\Funcf(\SetS) = \Funcf(\OptTree)$, proving the lemma in the base case.

    \paragraph*{Inductive case \texorpdfstring{($\RecLevel \geq 1$)}{(\textit{l} >= 1)}.}
    In this case, $\SetS$ is constructed in the Line~\ref{line:RG inductive case} block. Assume for induction that the lemma holds for every recursion level below $\RecLevel$.

    Consider the $\IterInd$th iteration of the Line~\ref{line:RG update loop} loop. Then, by Corollary~\ref{cor:target tree decomp}, $\OptTree$ can be partitioned in a way that is specific to the $\IterInd$th iteration. That is, $\OptTree$ can be partitioned into subtrees $\OptTree_{\IterInd, 1}, \dots, \OptTree_{\IterInd, \OptNumSub(\IterInd)}$ such that for each $\OptSubInd \in [1, \OptNumSub(\IterInd)] \colon \TargetSubsizeMin \leq \NumEdges(\OptTree_{\IterInd,\OptSubInd}) \leq \TargetSubsizeMax$.
    
    Now let $\BestOptSub$ be the subtree amongst $\OptTree_{\IterInd, 1}, \dots, \OptTree_{\IterInd, \OptNumSub(\IterInd)}$ maximizing $\frac{1}{\BestSubsize} \Funcf( \BestOptSub \mid \SetS_{\leq \IterInd-1} )$, where $\BestSubsize = \NumEdges( \BestOptSub )$. By a standard analysis, the maximality of $\BestOptSub$, and submodularity, Ineq.~\eqref{eqn:greedy density target subtree} below holds.
    \begin{align}
        \frac{1}{\BestSubsize} \Funcf(\BestOptSub \mid \SetS_{\leq \IterInd-1}) \geq \frac{1}{\TargetSize} \left( \Funcf(\OptTree) - \Funcf(\SetS_{\leq \IterInd-1}) \right). \label{eqn:greedy density target subtree}
    \end{align}
    Recall that $\BestSubsize = \NumEdges( \BestOptSub)$ and let $\Bestw$ be the root of $\BestOptSub$. Observe that $\BestOptSub$ is a feasible out-tree for the instance $(\Graph, \Funcf( \, \cdot \mid \SetS_{\leq \IterInd-1}), \BestSubsize, \Bestw)$ of \DRCSMa{}. Also, there is a recursive call in the inner Line~\ref{line:RG guess w and c} loop which solves this instance since \RA{} must guess $\Vertw = \Bestw$ and $\TargetSubsize = \BestSubsize$.
    Moreover, this recursive call has recursion level at most $\RecLevel-1$ since $\TargetSubsize \leq \frac{\TargetSize}{\TargetDiv} \leq \TargetDiv^{\RecLevel-1}$.
    Therefore, by induction, this recursive call outputs the out-subtree $\SetS^{\Bestw, \BestSubsize}_{\IterInd}$ satisfying
    \begin{align}
        \Funcf( \SetS^{\Bestw, \BestSubsize}_{\IterInd} \mid \SetS_{\leq \IterInd - 1}) &\geq \frac{1}{\RecLevel} \Funcf( \BestOptSub \mid \SetS_{\leq \IterInd - 1}). \label{eqn:recur subtree approx factor}
    \end{align}
    In the $\IterInd$th Line~\ref{line:RG update sol} update, \RA{} actually adds the out-subtree with maximum marginal density along with its connecting path, namely $\SetS^{\Vertw_{\IterInd}, \TargetSubsize_{\IterInd}}_{\IterInd} \cup \SetP_{\IterInd}$. Hence, we prove Ineq.~\eqref{eqn:recur partial sol gap} for all iterations $\IterInd$ below; the 1st inequality holds by the monotonicity of $\Funcf$, the 2nd by the maximality of $\SetS^{\Vertw_{\IterInd}, \TargetSubsize_{\IterInd}}_{\IterInd}$, the 3rd by Ineq.~\eqref{eqn:recur subtree approx factor}, and the 4th by Ineq.~\eqref{eqn:greedy density target subtree}.
    \begin{align}
        \frac{1}{\TargetSubsize_{\IterInd}}  \Funcf(\SetS^{\Vertw_{\IterInd}, \TargetSubsize_{\IterInd}}_{\IterInd}  \cup \SetP_{\IterInd} \mid \SetS_{\leq \IterInd-1})
        &\geq \frac{1}{\TargetSubsize_{\IterInd}}  \Funcf(\SetS^{\Bestw, \BestSubsize}_{\IterInd} \mid \SetS_{\leq \IterInd-1}) \notag \\
        &\geq \frac{1}{\BestSubsize} \Funcf( \SetS^{\Bestw, \BestSubsize}_{\IterInd} \mid \SetS_{\leq \IterInd - 1}) \notag \\
        &\geq \frac{1}{\BestSubsize \RecLevel} \Funcf( \BestOptSub \mid \SetS_{\leq \IterInd - 1}) \notag \\
        &\geq \frac{ 1 }{\TargetSize \RecLevel} \left( \Funcf(\SetT) - \Funcf(\SetS_{\leq \IterInd-1}) \right), \notag \\ 
        \Funcf(\SetT) - \Funcf(\SetS_{\leq\IterInd}) &\leq \left( 1 - \frac{\TargetSubsize_{\IterInd}}{\TargetSize \RecLevel} \right) \left( \Funcf(\SetT) - \Funcf(\SetS_{\leq\IterInd-1}) \right). \label{eqn:recur partial sol gap}
    \end{align}

    Since \RA{} makes $\NumSubtrees$ updates to construct $\SetS = \SetS_{\leq \NumSubtrees}$, we chain Ineq.~\eqref{eqn:recur partial sol gap} $\NumSubtrees$ times below. The 2nd inequality holds as $\prod_{\IterInd=1}^{\NumSubtrees} ( 1 - \frac{\TargetSubsize_{\IterInd}}{\TargetSize \RecLevel} )$ is maximized when $\TargetSubsize_{\IterInd} = \frac{\TargetSize}{ \NumSubtrees}$ for all $\IterInd$, noting that $\sum_{\IterInd = 1}^{\NumSubtrees} \TargetSubsize_{\IterInd} \leq \TargetSize$ is enforced by the Line~\ref{line:RG update loop} loop; Ineq.~\eqref{eqn:recur final sol gap} then holds by the non-negativity of $\Funcf$.
    \begin{align}
        \Funcf(\SetT) - \Funcf(\SetS_{\leq\NumSubtrees}) &\leq \left( \prod_{\IterInd=1}^{\NumSubtrees} \left( 1 - \frac{\TargetSubsize_{\IterInd}}{\TargetSize \RecLevel} \right) \right) \left( \Funcf(\SetT) - \Funcf(\SetS_{\leq 0}) \right) \notag \\
        &\leq \left( 1 - \frac{1}{\NumSubtrees \RecLevel} \right)^{\NumSubtrees} \left( \Funcf(\SetT) - \Funcf(\SetS_{\leq 0}) \right) \notag \\
        &\leq \left( 1 - \frac{1}{\NumSubtrees \RecLevel} \right)^{\NumSubtrees} \Funcf( \SetT ). \label{eqn:recur final sol gap}
    \end{align}
    Finally, we upper bound $\left( 1 - \frac{1}{\RecLevel \NumSubtrees} \right)^{\NumSubtrees}$ in the right-hand side of Ineq.~\eqref{eqn:recur final sol gap} using the inequality $(1-x)^{\NumSubtrees} < \frac{1}{1+\NumSubtrees x}$, which holds for all $x \in (-1, 1)$ and $\NumSubtrees \geq 0$. From there, rearranging gives the lemma inequality.
    \begin{align}
        \Funcf(\SetT) - \Funcf(\SetS_{\leq\NumSubtrees}) &\leq \frac{1}{1+\frac{\NumSubtrees}{\RecLevel \NumSubtrees}} \Funcf(\SetT) \notag \\
        &= \frac{\RecLevel}{\RecLevel+1} \Funcf(\SetT), \notag \\
        \Funcf(\SetS_{\leq \NumSubtrees}) &\geq \frac{1}{\RecLevel+1} \Funcf(\SetT). \tag*{\qedhere}
    \end{align}
\end{proof}

\paragraph*{Violation Factor of \RA{}.}
Before bounding the size of the output out-tree in Lemma~\ref{lem:RG sol size}, we bound the number of updates $\NumSubtrees$ to the output out-tree. 
\begin{claim} \label{claim:recur num of updates UB}
    $\NumSubtrees \leq 3 \TargetDiv + \log_{3/2} ( \frac{\TargetSize}{\TargetDiv} )$.
\end{claim}
\begin{proof}
    Let $\NumSubtrees_{1}$ be the number of iterations of the Line~\ref{line:RG update loop} loop where $\TargetSize \geq \TargetSize_{\IterInd-1} > \frac{\TargetSize}{\TargetDiv}$ holds, and $\NumSubtrees_{2}$ be the number of remaining iterations where $\frac{\TargetSize}{\TargetDiv} \geq \TargetSize_{\IterInd-1} > 0$ holds.
    
    In each of the first $\NumSubtrees_{1}$ iterations, \RA{} decrements $\TargetSize_{\IterInd-1}$ by at least $\TargetSubsizeMin \geq \frac{\TargetSize}{3 \TargetDiv}$ to give $\TargetSize_{\IterInd}$. Thus, it must hold that $\TargetSize - (\NumSubtrees_{1}-1) \frac{\TargetSize}{3 \TargetDiv} > \frac{\TargetSize}{\TargetDiv}$. Solving for $\NumSubtrees_{1}$ gives $\NumSubtrees_{1} < 3\TargetDiv - 2$.

    In each of the remaining $\NumSubtrees_{2}$ iterations, \RA{} decrements $\TargetSize_{\IterInd-1}$ by at least $\TargetSubsizeMin \geq \frac{\TargetSize_{\IterInd-1}}{3}$  to give $\TargetSize_{\IterInd}$, with the last iteration possibly having $\TargetSize_{\IterInd-1} = 1$. Thus, it must hold that $\left(\frac{2}{3}\right)^{(\NumSubtrees_{2} - 2)} \frac{\TargetSize}{\TargetDiv} \geq 1$. Solving for $\NumSubtrees_{2}$ gives $\NumSubtrees_{2} \leq \log_{3/2}(\frac{\TargetSize}{\TargetDiv}) + 2$.
    
    We now have $\NumSubtrees = \NumSubtrees_{1} + \NumSubtrees_{2} \leq 3 \TargetDiv + \log_{3/2} ( \frac{\TargetSize}{\TargetDiv} )$, proving the claim.
\end{proof}

\begin{lemma} \label{lem:RG sol size}
    Let $(\Graph, \Funcf, \TargetSize, \Vertv)$ be an instance of \DRCSMa{}. Let $\TargetDiv > 1$ be the size divisor given to \RA{}, and $\RecLevel$ be the integer satisfying $\TargetDiv^{\RecLevel-1} < \TargetSize \leq \TargetDiv^{\RecLevel}$. Then \RA{} outputs an out-tree $\SetS$ satisfying
    \begin{align} 
        \NumEdges(\SetS) \leq \left(1 + 3\RecLevel\TargetDiv + \RecLevel \log_{\frac{3}{2}} \left( \frac{\TargetSize}{\TargetDiv} \right) \right) \TargetSize. \notag
    \end{align}
\end{lemma}
\begin{proof}
    Let $\SolSizeFunc(\TargetSize)$ be the function that gives the worst-case number of edges in an out-tree $\SetS$ output by \RA{}, given a size constraint of $\TargetSize$.

    By the construction of $\SetS$, the recurrence relation below holds for $\SolSizeFunc(\TargetSize)$. We explain this recurrence below. Recall that $\TargetSubsize_{\IterInd}$ denotes the size constraint used to recursively construct the $\IterInd$th out-subtree added to $\SetS$.
    \begin{itemize}
        \item The $\TargetSize = 1$ case of Eq.~\eqref{eqn:RG sol size rec} holds since, for $\TargetSize=1$, \RA{} finds a solution of size 1 (Line~\ref{line:RG base case update sol}).
        
        \item The $\TargetSize \geq 2$ case of Eq.~\eqref{eqn:RG sol size rec} holds since, for $\TargetSize \geq 2$, \RA{} finds a solution by adding $\NumSubtrees$ subtrees and connecting paths (Line~\ref{line:RG update sol}). The $\IterInd$th added subtree has size at most $ \SolSizeFunc(\TargetSubsize_{\IterInd})$ (Line~\ref{line:RG find w tree}). Moreover, the path $\SetP_{\IterInd}$ connecting $\Vertv$ to the $\IterInd$th out-subtree has length at most $\TargetSize$ (Line~\ref{line:RG set of rec centers}).
        
        \item Constraint~\eqref{eqn:RG num conn paths} holds by Claim~\ref{claim:recur num of updates UB} and the fact that $\NumSubtrees \leq \TargetSize$.
        
        \item Constraint~\eqref{eqn:RG target subsize} holds by the assignments of $\TargetSubsizeMin$ and $\TargetSubsizeMax$ in Lines~\ref{line:RG assign target subsize min} and \ref{line:RG assign target subsize max}, and since, for each $\IterInd$, $\TargetSubsizeMin \leq \TargetSubsize_{\IterInd} \leq \TargetSubsizeMax$.
        \item Constraint~\eqref{eqn:RG sum of target subsizes} holds since the Line~\ref{line:RG update loop} loop only terminates once $\TargetSubsize_{1} + \dots + \TargetSubsize_{\NumSubtrees} = \TargetSize$.
    \end{itemize}
    \begin{align}
        \SolSizeFunc(\TargetSize) &= \begin{cases}
        \begin{aligned}
            &1, &&\TargetSize = 1 \\
            &\sum_{\IterInd=1}^{\NumSubtrees} ( \SolSizeFunc( \TargetSubsize_{\IterInd} ) + \TargetSize ), &&\TargetSize \geq 2
        \end{aligned} \ ,
        \end{cases} \label{eqn:RG sol size rec} \\
        \text{where} \notag \\
        1 \leq \NumSubtrees &\leq \min\left\{ 3 \TargetDiv + \log_{\frac{3}{2}} \left( \frac{\TargetSize}{\TargetDiv} \right), \TargetSize \right\}, \label{eqn:RG num conn paths} \\
        \forall \IterInd \in [1, \NumSubtrees] \colon 1 \leq \TargetSubsize_{\IterInd} &\leq \max \left\{ 1, \frac{\TargetSize}{\TargetDiv} \right\}, \label{eqn:RG target subsize} \\
        \sum_{\IterInd=1}^{\NumSubtrees} \TargetSubsize_{\IterInd} &= \TargetSize. \label{eqn:RG sum of target subsizes}
    \end{align}     
    
    Now we prove the lemma by upper bounding $\SolSizeFunc(\TargetSize)$ and, therefore, $\NumEdges(\SetS)$. We prove Ineq.~\eqref{eqn:RG sol size UB} below by induction on the recursion level, which is the integer $\RecLevel$ satisfying $\TargetDiv^{\RecLevel-1} < \TargetSize \leq \TargetDiv^{\RecLevel}$. Note that Ineq.~\eqref{eqn:RG sol size UB} clearly holds when $\RecLevel = 0$ since $\TargetSize \leq \TargetDiv^{\RecLevel} = 1$ and this means that $\SolSizeFunc(\TargetSize) = 1$ according to Eq.~\eqref{eqn:RG sol size rec}.
    \begin{align}
        \SolSizeFunc(\TargetSize) &\leq (1 + \RecLevel \NumSubtrees) \TargetSize. \label{eqn:RG sol size UB}
    \end{align}

    \paragraph*{Base Case \texorpdfstring{($\RecLevel=1$)}{(\textit{l} = 1)}.}
    In this case, $1 < \TargetSize \leq \TargetDiv$, so $\frac{\TargetSize}{\TargetDiv} \leq 1$. Then, by Constraint~\eqref{eqn:RG target subsize}, it holds that for all iterations $\IterInd \in [1, \NumSubtrees] \colon \TargetSubsize_{\IterInd} = 1 $. Therefore, by the $\TargetSize = 1$ case of Eq.~\eqref{eqn:RG sol size rec}, it holds that $\SolSizeFunc(\TargetSubsize_{\IterInd}) = 1$. Thus, we simplify the $\TargetSize \geq 2$ case of Eq.~\eqref{eqn:RG sol size rec} below; the inequality holds by the bound $\NumSubtrees \leq \TargetSize$ as implied by Constraint~\eqref{eqn:RG num conn paths}.
    \begin{align}
        \SolSizeFunc(\TargetSize) &= \sum_{\IterInd=1}^{\NumSubtrees} \SolSizeFunc(\TargetSubsize_{\IterInd}) + \NumSubtrees \TargetSize
        = \NumSubtrees + \NumSubtrees \TargetSize
        \leq \TargetSize + \NumSubtrees \TargetSize
        = (1 + \NumSubtrees) \TargetSize. \notag
    \end{align}

    \paragraph*{Inductive Case \texorpdfstring{($\RecLevel \geq 2$)}{(\textit{l} >= 2)}.} Assume for induction that Ineq.~\eqref{eqn:RG sol size UB} holds for every recursion level below $\RecLevel$. By Constraint~\eqref{eqn:RG target subsize} and $\TargetSize \leq \TargetDiv^{\RecLevel}$, we have that for all iterations $\IterInd \in [1, \NumSubtrees] \colon 1 \leq \TargetSubsize_{\IterInd} \leq \frac{\TargetSize}{\TargetDiv} \leq \TargetDiv^{\RecLevel-1} $. Thus, we can bound each $\SolSizeFunc(\TargetSubsize_{\IterInd})$ by induction. Using this, we simplify the $\TargetSize \geq 2$ case of Eq.~\eqref{eqn:RG sol size rec}; the 3rd equality holds by Constraint~\eqref{eqn:RG sum of target subsizes} since it requires that $\sum_{\IterInd=1}^{\NumSubtrees} \TargetSubsize_{\IterInd} = \TargetSize$.
    \begin{align}
        \SolSizeFunc(\TargetSize) &= \sum_{\IterInd=1}^{\NumSubtrees} \SolSizeFunc(\TargetSubsize_{\IterInd}) + \NumSubtrees \TargetSize
        \leq \sum_{\IterInd=1}^{\NumSubtrees} (1 + (\RecLevel-1) \NumSubtrees ) \TargetSubsize_{\IterInd} + \NumSubtrees \TargetSize \notag \\
        &= (1 + (\RecLevel-1) \NumSubtrees ) \sum_{\IterInd=1}^{\NumSubtrees}  \TargetSubsize_{\IterInd} + \NumSubtrees \TargetSize
        = (1 + (\RecLevel-1) \NumSubtrees) \TargetSize + \NumSubtrees \TargetSize \notag \\
        &= (1 + \RecLevel \NumSubtrees) \TargetSize. \notag
    \end{align}
    
    We have proved Ineq.~\eqref{eqn:RG sol size UB} in both the base case and inductive case. Finally, substituting $\NumSubtrees \leq 3 \TargetDiv + \log_{3/2} \left( \frac{\TargetSize}{\TargetDiv} \right)$ from Constraint~\eqref{eqn:RG num conn paths} into Ineq.~\eqref{eqn:RG sol size UB} proves the lemma.
\end{proof}

\paragraph*{Running Time of \RA{}.}
\begin{lemma} \label{lem:RG running time}
    At recursion level $\RecLevel$, \RA{} runs in time $ O(\NumVertices^{\RecLevel+1} \TargetSize^{2\RecLevel+2})$.
\end{lemma}
\begin{proof}
    We prove the lemma by induction on $\RecLevel$. Consider the base case $\RecLevel = 0$. It holds that $\TargetSize \leq \TargetDiv^\RecLevel = 1$, so the Line~\ref{line:RG base case} block is run. In this block, there are at most $\NumVertices$ value oracle queries. Thus, the total running time is $\NumVertices = O(\NumVertices^{\RecLevel+1} \TargetSize^{2\RecLevel+2})$.
    
    Now consider the inductive case $\RecLevel \geq 1$. In each iteration of the Line~\ref{line:RG update loop} loop, there are $\NumVertices \TargetSize$ recursive calls to \RA{} as well as $\NumVertices \TargetSize$ value oracle queries in Line~\ref{line:RG select subtree greedily}. Each recursive call has recursion level $\RecLevel-1$, so by induction has running time $ O(\NumVertices^{\RecLevel} (\frac{\TargetSize}{\TargetDiv})^{2\RecLevel}) = O(\NumVertices^{\RecLevel} \TargetSize^{2\RecLevel})$. Thus, the running time of each iteration of the Line~\ref{line:RG update loop} loop is $\NumVertices \TargetSize (O(\NumVertices^{\RecLevel} \TargetSize^{2\RecLevel}) + 1) = O(\NumVertices^{\RecLevel+1} \TargetSize^{2\RecLevel+1}) + \NumVertices \TargetSize$. Lastly, there are at most $\TargetSize$ iterations of the Line~\ref{line:RG update loop} loop. Thus, the total running time is $\TargetSize (O(\NumVertices^{\RecLevel+1} \TargetSize^{2\RecLevel+1}) + \NumVertices \TargetSize) = O(\NumVertices^{\RecLevel+1} \TargetSize^{2\RecLevel+2}) + \NumVertices \TargetSize^2 = O(\NumVertices^{\RecLevel+1} \TargetSize^{2\RecLevel+2})$.
\end{proof}

\subsection{Analysis of \texorpdfstring{\RAd{}}{\RA-\textit{d}}} \label{sec:RAd analysis}
Recall that \RAd{} simply takes an instance $(\Graph, \Funcf, \Card, \Vertv)$ of \DRCSMa{} and outputs the solution from \RA$(\Graph, \Funcf, \Card, \Vertv, \Card^{\frac{1}{\RecDepth}})$. We state the performance guarantees of \RAd{} in Theorem~\ref{thm:RAd}, which we prove after proving a few simple claims.

\begin{theorem}\label{thm:RAd}
    Let $(\Graph, \Funcf, \Card, \Vertv)$ be an instance of \DRCSMa{} and $\RecDepth \geq 1$ be an integer. Then \RAd{} is a $(\frac{1}{\RecDepth+1}, (\RecDepth+1)^2 \Card^{\frac{1}{\RecDepth}})$-approximation algorithm for \DRCSMa{} that runs in time $O(\NumVertices^{\RecDepth+1} \Card^{2\RecDepth+2})$.
\end{theorem}

We first prove Claim~\ref{claim:rec level of RG call}, which concerns the recursion level of the initial call to \RA{}. Recall that the recursion level (of a given \RA{} call) is the integer $\RecLevel$ satisfying $\TargetDiv^{\RecLevel-1} < \TargetSize \leq \TargetDiv^{\RecLevel}$.
Afterwards, we prove Claims~\ref{claim:ln b UB} and \ref{claim:RG final sol size}, which are relevant to bounding the size of the tree output by \RA{}.

\begin{claim} \label{claim:rec level of RG call}
    Suppose \RA{} is called with size constraint $\TargetSize = \Card$ and size divisor $\TargetDiv = \Card^{\frac{1}{\RecDepth}}$ in its input. Then its recursion level is $\RecLevel = \RecDepth$.
\end{claim}
\begin{proof}
    Since $\TargetSize = \Card$ and $\TargetDiv = \Card^{\frac{1}{\RecDepth}}$, we have that $\TargetDiv^{\RecDepth - 1} = \Card^{\frac{\RecDepth-1}{\RecDepth}} \leq \TargetSize = \Card^{\frac{\RecDepth}{\RecDepth}} = \TargetDiv^{\RecDepth}$.
\end{proof}

\begin{claim} \label{claim:ln b UB}
    For all $\RecDepth > 0$ and $\Card > 0 \colon \ln(\Card) \leq \RecDepth \Card^{\frac{1}{\RecDepth}} / e $.
\end{claim}
\begin{proof}
    Begin from the inequality $\ln y \leq y / e$, and substitute $y = \Card^{\frac{1}{\RecDepth}} $ to get $\frac{1}{\RecDepth} \ln \Card \leq \TargetSize^{ \frac{1}{\RecDepth} } / e$. Rearranging this proves the claim.
\end{proof}

\begin{claim} \label{claim:RG final sol size}
    Suppose \RA{} is called with size constraint $\TargetSize = \Card$ and size divisor $\TargetDiv = \Card^{\frac{1}{\RecDepth}}$ in its input. Then \RA{} outputs an out-tree with $\NumEdges( \SetS ) \leq (\RecDepth+1)^2 \Card^{\frac{1}{\RecDepth}} \Card$.
\end{claim}

\begin{proof}
    By the assignments of $\TargetSize = \Card$ and $\TargetDiv = \Card^{\frac{1}{\RecDepth}}$ and by Claim~\ref{claim:rec level of RG call}, the recursion level is $\RecLevel = \RecDepth$. By substituting these values into Theorem~\ref{thm:RA}, we upper bound $\NumEdges(\SetS)$ below, proving the claim; the 2nd inequality holds by Claim~\ref{claim:ln b UB}.
    \begin{align}
        \NumEdges(\SetS)
        &\leq \left(1 + 3\RecLevel\TargetDiv + \RecLevel \log_{\frac{3}{2}} \left( \frac{\Card}{\TargetDiv} \right) \right) \Card
        = \left( 1 + 3 \RecDepth \Card^{\frac{1}{\RecDepth}} + \RecDepth \log_{\frac{3}{2}} \left( \Card^{\frac{\RecDepth-1}{\RecDepth}} \right) \right) \Card \notag \\
        &= \left( 1 + 3 \RecDepth \Card^{\frac{1}{\RecDepth}} + \frac{\RecDepth-1}{\ln \frac{3}{2}} \ln ( \Card ) \right) \Card
        \leq \left( 1 + 3 \RecDepth \Card^{\frac{1}{\RecDepth}} + \frac{\RecDepth-1}{\ln \frac{3}{2}} \frac{\RecDepth \Card^{\frac{1}{\RecDepth}}}{e} \right) \Card \notag \\
        &\leq \left(\Card^{\frac{1}{\RecDepth}} + 3 \RecDepth \Card^{\frac{1}{\RecDepth}} + \RecDepth(\RecDepth-1) \Card^{\frac{1}{\RecDepth}} \right) \Card
        = \left(1 + 2\RecDepth + \RecDepth^2 \right) \Card^{\frac{1}{\RecDepth}} \Card \notag \\
        &= (\RecDepth+1)^{2} \Card^{\frac{1}{\RecDepth}} \Card. \tag*{\qedhere}
    \end{align}
\end{proof}

\begin{proof}[Proof of Theorem~\ref{thm:RAd}]
    \RAd{} finds the output out-tree $\SetS$ by calling \RA{} with size constraint $\TargetSize = \Card$ and size divisor $\TargetDiv = \Card^{\frac{1}{\RecDepth}}$. By Claim~\ref{claim:rec level of RG call}, the recursion level of this \RA{} call is $\RecLevel = \RecDepth$. Thus, it follows from Theorem~\ref{thm:RA} that $\SetS$ has value $\Funcf(\SetS) \geq \frac{1}{\RecDepth+1} \Funcf(\OptTree)$. Further, by Claim~\ref{claim:RG final sol size}, we have that $\NumEdges(\SetS) \leq (\RecDepth+1)^{2} \Card^{\frac{1}{\RecDepth}} \Card$. Hence, \RAd{} is a $(\frac{1}{\RecDepth+1}, (\RecDepth+1)^2 \Card^{\frac{1}{\RecDepth}})$-approximation algorithm for \DRCSMa{}.

    Finally, by Theorem~\ref{thm:RA}, $\TargetSize = \Card$, and the recursion level being $\RecLevel = \RecDepth$, the running time of \RAd{} is $O(\NumVertices^{\RecLevel+1} \TargetSize^{2\RecLevel+2}) = O(\NumVertices^{\RecDepth+1} \Card^{2\RecDepth+2})$. This proves the lemma.
\end{proof}

\section{Related Work} \label{sec:related work}
\paragraph*{Variants of (Directed Rooted) \CSM{}}
Kuo et al. \cite{kuo2014maximizing} introduced the variant of \CSMa{} known as \emph{Maximum Connected Submodular function with a Budget constraint} (\MCSBa). Here, we are given an undirected graph $\Graph$ with positive integer vertex costs, a non-negative monotone submodular function $\Funcf$ on the vertices of $\Graph$, and a budget $\Budget$. The goal is to select a connected subgraph of $\Graph$, of cost at most $\Budget$, whose vertex set maximizes $\Funcf$. Kuo et al. \cite{kuo2014maximizing} showed a polynomial time $\Omega(\frac{1}{(\Degree+1) \sqrt{\Budget}})$-approximation algorithm for \MCSBa{}, where $\Degree$ is the maximum degree in $\Graph$. D'Angelo et al. \cite{DAngelo2022} later improved on this by giving a polynomial time $\Omega(\frac{1}{\sqrt{\Budget}})$-approximation algorithm for the directed generalization of \MCSBa{} known as \emph{Directed Unrooted Submodular Tree} (\DUSTa{}).

Ghuge and Nagarajan \cite{Ghuge2022} introduced the variant of \DRCSMa{} known as \emph{Submodular Tree Orienteering} (\STOa{}). Here, we are given a directed graph $\Graph$ with positive integer edge costs, a non-negative monotone submodular function $\Funcf$, a budget $\Budget$, and a root vertex $\Vertv$. The goal is to find an out-tree in $\Graph$, with root $\Vertv$ and cost at most $\Budget$, whose vertex set maximizes $\Funcf$. Ghuge and Nagarajan showed a quasi-polynomial time $\Omega(\frac{\log \log \NumOptVertices}{\log \NumOptVertices})$-approximation algorithm, where $\NumOptVertices$ is the number of vertices in an optimal tree. This approximation is tight as implied by a known hardness result for Group Steiner Tree (see below).

D'Angelo et al. \cite{DAngelo2022} introduced the variant of \DRCSMa{} known as \emph{Directed Rooted Submodular Tree} (\DRSTa{}). \DRSTa{} is defined like \STOa{} except that positive integer costs are assigned to vertices rather than edges; in fact, D'Angelo et al. showed an approximation-preserving reduction from \STOa{} to \DRSTa{}. D'Angelo et al. \cite{DAngelo2022} showed, for every $\Errorv \in (0, 1]$, a polynomial time $O(\frac{\Errorv^{3}}{\sqrt{\Budget}})$-approximation algorithm that violates the budget by at most a factor of $(1+\Errorv)$.

\paragraph*{Variants of (Directed Rooted) \CMC{}}
Ran et al. \cite{Ran2016} introduced the variant of \CMCa{} known as \emph{Connected Budgeted maximum Coverage} (\CBCa{}) (they called this problem \emph{Maximum Weight Budgeted Connected Set Cover}). Here, we are given a universe of elements $\Univ$ wherein each element has a non-negative weight, an undirected graph $\Graph$ where each vertex is a subset of $\Univ$ and has a non-negative cost, and a budget $\Budget$. The goal is to select a connected subgraph of $\Graph$, of cost at most $\Budget$, whose vertex set maximizes the total weight of its covered elements. Ran et al. studied \CBCa{} under the assumption that every pair of subsets (vertices) with a non-empty intersection is adjacent in $\Graph$, which is a realistic assumption in, e.g., wireless sensor network applications. They gave a polynomial time $\Omega(\frac{1}{(\Degree + 1) \log \NumVertices})$-approximation algorithm, where $\Degree$ is the maximum degree in $\Graph$ and $\NumVertices$ is the number of vertices in $\Graph$.
D'Angelo et al. \cite{DAngelo2022} later improved the approximation factor to $\Omega(\frac{1}{\log \NumVertices})$ under the same assumption.
D'Angelo and Delfaraz \cite{DAngelo2025} studied \CBCa{} in general; they showed, for every $\Errorv \in (0,1]$, a polynomial time $\Omega(\frac{\Errorv^2}{\log(\NumVertices + |\Univ|) \log |\Univ|})$-approximation algorithm that violates the budget by at most a factor of $(1 + \Errorv)$.

D'Angelo and Delfaraz \cite{DAngelo2025} also introduced and studied the variant of \DRCMCa{} known as \emph{Directed rooted Connected Budgeted maximum Coverage} (\DCBCa{}). \DCBCa{} generalizes \CBCa{} by letting $\Graph$ be a directed graph and by specifying a root vertex $\Vertv$; the goal is to select an out-tree, with root $\Vertv$ and cost at most $\Budget$, whose vertex set maximizes the total weight of its covered elements. D'Angelo and Delfaraz \cite{DAngelo2025} gave, for every $\Errorv \in (0,1]$, a polynomial time $\Omega(\frac{\Errorv^2}{\sqrt{\NumVertices} \log^{2} |\Univ|})$-approximation algorithm that violates the budget by at most a factor of $(1 + \Errorv)$. They mentioned that their techniques cannot be directly applied to the generalizations of \CBCa{} and \DCBCa{} where the objective function is monotone submodular.

We briefly mention that previous works have studied \CMCa{} in special geometric settings, achieving polynomial time constant-factor approximations \cite{Huang2015, Huang2019}. Also, previous works have studied the special case of \CMCa{} known as \emph{\BCDS{}}, again achieving polynomial time constant-factor approximations \cite{Avrachenkov2014, Khuller2014, Lamprou2021}. 
This problem has application to selecting a limited set of connected and influential users in a social network.

\paragraph*{Group Steiner Tree}  
In \emph{Group Steiner Tree} (\GSTa{}), we are given an undirected graph $\Graph$ with edge costs and a collection of subsets of the vertices, called \emph{groups}. The goal is to select a minimum cost tree in $\Graph$ that contains at least one vertex from every group. This problem is polylogarithmic approximable by both polynomial time \cite{Garg2000, Charikar1998, Zosin2002, Chekuri2006} and quasi-polynomial time \cite{Charikar1999, Grandoni2019} algorithms; in particular Grandoni et al. \cite{Grandoni2019} gave a quasi-polynomial time $O(\frac{\log^{2} \NumVertices}{ \log \log \NumVertices})$-approximation algorithm and showed a matching hardness result by modifying the hardness result of Halperin and Krauthgamer \cite{Halperin2003}.

\GSTa{} generalizes \emph{Set Cover} and is a special case of \emph{Directed Steiner Tree} (\DSTa{}). Interestingly, \DSTa{} has so far appeared to be more challenging than \GSTa{} for polynomial time algorithms. Charikar et al. \cite{Charikar1999} gave the current-best polynomial time algorithm for \DSTa{}, achieving a $O(\frac{\NumTerm^{\Error} \log \NumTerm}{\Error^2})$-approximation in polynomial time for a constant $\Error > 0$, where $\NumTerm$ is the number of terminal vertices to be covered.

\paragraph*{Hardness Results for \GSTa{} and Budgeted Variants of \CSMa{}}
For \GSTa{}, there are known approximation hardness results that apply to quasi-polynomial time algorithms and, thus, to polynomial time algorithms as well. Letting $\NumVertices$ be the input size, Halperin and Krauthgamer \cite{Halperin2003} showed that, for every fixed $\Error > 0$, a $O( \log^{2-\Error} \NumVertices )$-approximation is not possible in quasi-polynomial time assuming $ \text{NP} \nsubseteq \text{ZPTIME}(n^{\polylog n}) $. This result holds even when the input graph is restricted to being a tree, in which case edge and vertex costs are equivalent. Grandoni et al. \cite{Grandoni2019} modified this hardness result to show that a $o(\frac{\log^{2} \NumVertices}{ \log \log \NumVertices})$-approximation is not possible in quasi-polynomial time assuming the Projection Game Conjecture and $ \text{NP} \nsubseteq \bigcap_{0 < \Error < 1} \text{ZPTIME}( 2^{n^{\Error}} ) $. They also show that this is tight by providing a quasi-polynomial time $O(\frac{\log^{2} \NumVertices}{ \log \log \NumVertices})$-approximation algorithm.

For variants of \CSMa{} with edge or vertex costs, there are near-logarithmic approximation hardness results as implied by the above hardness results for \GSTa{} \cite{Halperin2003, Grandoni2019}; we explain these implications below. We explain how the hardness results for \GSTa{} and the budgeted variant of \CSMa{} (with edge costs). Consider a maximization variant of \GSTa{}, where we are additionally given a budget $\Budget$, and the goal now is to select a tree in $\Graph$, of cost at most $\Budget$, that maximizes the number of groups it covers. This is just a special case of budgeted \CSMa{} with edge costs, since the objective function to be maximized is a coverage function. Then, for a value $\alpha \geq 1$, every $\frac{1}{\alpha}$-approximation algorithm for budgeted \CSMa{} implies an $(\alpha \ln \NumVertices)$-approximation algorithm for \GSTa{} by a set-covering approach. Thus, the aforementioned hardness results imply that, for budgeted \CSMa{}, $\Omega( \frac{1}{\log^{1-\Error} \NumVertices} )$ and $\omega(\frac{\log \log \NumVertices}{\log \NumVertices})$-approximations are not possible in quasi-polynomial time (given the respective hardness assumptions). These hardness results extend to the budgeted variants of (Directed Rooted) \CSMa{} mentioned above, including those with vertex costs rather than edge costs.
However, they do not carry over easily to \CSMa{} (which has uniform costs) since the original hardness results for \GSTa{} rely on constructing a tree instance with carefully assigned edge costs.

We also mention that, for the vertex-cost variant of \CSMa{}, Kortsarz and Nutov \cite{kortsarz2011approximating} showed that, even when allowing the budget to be violated by a universal constant factor, a $\omega(\frac{1}{\log \log \NumVertices })$-approximation is not possible in quasi-polynomial time assuming $\text{NP} \nsubseteq \text{Quasi}(\text{P})$. In fact, Kortsarz and Nutov proved their hardness result for the special case of the problem where the objective function is additive.

\paragraph*{Polymatroid Steiner Tree} In \emph{\PST} (\PSTa{}), we are given an undirected graph $\Graph$ with edge costs and a non-negative monotone integer-valued submodular function, $\Funcf$, on the vertex set of $\Graph$. The goal is to select a minimum cost tree $\SetS \subseteq \Graph$ with $\Funcf(\SetS) = \Covf$, i.e., the vertex set of $\SetS$ is a base of the polymatroid associated with $\Funcf$. This problem generalizes \GSTa{} and Covering Steiner Tree.

\PSTa{} was introduced by Calinescu and Zelikovsky \cite{Calinescu2005}, motivated by an application to sensor networks. They gave a polynomial time $O( \log^{1+\Error} \NumVertices \log \NumTerm )$-approximation algorithm, where $\NumTerm = \Funcf(\VertexSet)$, for the case where $\Graph$ is a tree; this algorithm extends the combinatorial algorithm of Chekuri et al. \cite{Chekuri2006} for \GSTa{}. For \PSTa{} on general undirected graphs, a polynomial time $O( \log^{2+\Error} \NumVertices \log \NumTerm )$-approximation is obtained by taking a tree embedding of the input graph and running the above algorithm on it.

Im et al. showed \cite{im2016minimum} that the above algorithm implies a bicriteria $(\Omega(1), O(\log^{2+\Error} \NumVertices))$-approximation for \SO{}; the same implication holds for \STOa{} on undirected graphs as this is essentially equivalent to Submodular Orienteering. This further implies a polynomial time $\Omega(\frac{1}{\log^{2+\Error} \NumVertices})$-approximation for \CSMa{} by selecting a valuable size-$\Card$ subtree.

\paragraph*{Polymatroid Directed Steiner Tree}
In \emph{\PDST} (\PDSTa{}), we are given a directed graph $\Graph$ with edge costs, a non-negative monotone integer-valued submodular function, $\Funcf$, on the vertex set of $\Graph$, and a root vertex $\Vertv$. The goal is to select a minimum cost out-tree $\SetS \subseteq \Graph$ with root $\Vertv$ and $\Funcf(\SetS) = \Covf$.

\PDSTa{} was also introduced by Calinescu and Zelikovsky \cite{Calinescu2005}. They gave, for each integer $\RecDepth \geq 2$, a $\RecDepth^{3} \NumTerm^{\frac{1}{\RecDepth}}$-approximation algorithm that runs in time $O(\NumVertices^{\RecDepth+1} \NumTerm^{2\RecDepth+2})$, where $\NumTerm = \Funcf(\VertexSet)$; this algorithm extends the algorithm by Charikar et al. \cite{Charikar1999} for \DST{}. Ghuge and Nagarajan \cite{Ghuge2022} showed that their quasi-polynomial time algorithm for \STOa{} implies, by a set-covering approach, a quasi-polynomial time $O(\frac{\log^2 \NumTerm}{\log \log \NumTerm})$-approximation algorithm for \PDSTa{}. Recently, Chekuri et al. \cite{Chekuri2024} considered \PDSTa{}, as well as its special cases, when $\Graph$ is planar, showing polynomial time polylogarithmic approximation algorithms. Their work builds on ideas from Friggstad and Mousavi \cite{Friggstad2023} for \DST{} on planar graphs.

\paragraph*{\RAd{} Compared to Algorithm for Polymatroid Directed Steiner Tree}
Recall that, for each integer $\RecDepth \geq 1$, our algorithm \RAd{} (Algorithm~\ref{alg:RA}) finds a bicriteria $( \frac{1}{\RecDepth+1}, (\RecDepth+1)^2 \Card^{\frac{1}{\RecDepth}} )$-approximation for \DRCSMa{} in time $O(\NumVertices^{\RecDepth+1} \Card^{2 \RecDepth+2})$ (Theorem~\ref{thm:RAd}) by calling \RA{} with size constraint $\Card$, root $\Vertv$, and value $\TargetDiv = \Card^{\frac{1}{\RecDepth}}$. Similarly, for each integer $\RecDepth \geq 2$, Calinescu and Zelikovsky's algorithm finds a $\RecDepth^{3} \NumTerm^{\frac{1}{\RecDepth}}$-approximation for \PDSTa{} in time $O(\NumVertices^{\RecDepth+1} \NumTerm^{2\RecDepth+2})$, where $\NumTerm = \Funcf(\VertexSet)$. Calinescu and Zelikovsky's \PDSTa{} algorithm uses a recursive greedy strategy similar to our algorithm \RA{}. In general, their algorithm takes as input $(\Funcf, \RecLevel, \NumTerm, \Vertv, \SetX)$, where $\Funcf$ is the objective function, $\RecLevel$ is the recursion level, $\NumTerm$ is the target value, $\Vertv$ is the root vertex, and $\SetX$ is a set of already spanned vertices. Their algorithm constructs an out-tree $\SetS$ by adding cost-efficient out-subtrees (constructed recursively) until $\Funcf(\SetS \mid \SetX) \geq \NumTerm$. However, the main difference between their algorithm and \RA{} is that, for each solution update, they iterate over guesses, $\NumTermp$, of the marginal gain of a cost-efficient out-subtree (these guesses are passed to the recursive calls as the target values). This is why their algorithm's running time depends on $\NumTerm = \Covf$, which may not be polynomially bounded in $\GroundSetSize$. In \RA{} on the other hand, for each solution update, it iterates over guesses of the size of a valuable out-subtree, which are at most $\Card$.

We point out that it is possible to adapt the \PDSTa{} algorithm to achieve a bicriteria approximation for \DRCSMa{} similar to the bicriteria $( \frac{1}{\RecDepth+1}, (\RecDepth+1)^2 \Card^{\frac{1}{\RecDepth}} )$-approximation of \RAd{}. However, this requires additional steps to ensure that the \PDSTa{} algorithm runs in polynomial time and so that it returns an out-tree $\SetS$ with value $\Funcf(\SetS) \geq \Omega(\OptValue)$ and size $O(\Card^{\frac{1}{\RecDepth}}) \Card$ (for constant $\RecDepth$); this results in an extra polynomial factor in the running time compared to that of \RAd{}. Below we sketch our approach to achieving a bicriteria approximation for \DRCSMa{} using the \PDSTa{} algorithm, though we omit the details.

Let $(\Graph = (\VertexSet, \EdgeSet), \Funcg, \Card, \Vertv)$ be an instance of \DRCSMa{}. We first follow the approach of Chekuri and Pal \cite{Chekuri2005} to scale the objective function $\Funcg$ to give an integer-valued function $\Funcgp$ such that the scaled optimal solution value, $\OptValuep$, is polynomially bounded: we guess the most valuable vertex in the (original) optimal solution, $\OptVert$, and define $\Funcgp(\SetS) = \lfloor \frac{\Card^2 \Funcg(\SetS)}{\Error \Funcg(\OptVert)} \rfloor$, where $\Error > 0$ is any constant (the rounding does not guarantee that $\Funcgp$ is submodular, but this only loses a constant-factor in the overall approximation). It follows from subadditivity that the scaled optimal solution value, $\OptValuep$, is upper bounded by $\lfloor \frac{\Card^2 \OptValue}{\Error \Funcg(\OptVert)} \rfloor \leq \lfloor \frac{\Card^2 \cdot \Card \cdot \Funcg(\OptVert)}{\Error \Funcg(\OptVert)} \rfloor \leq \frac{\Card^{3}}{\Error}$.
Now we follow an approach similar to that in Theorem 3.1 of Im et al. \cite{im2016minimum} to solve the scaled instance: we guess the scaled optimal value, $\OptValuep$, and define the truncated function $\Funcf(\SetS) = \min\{ \OptValuep, \Funcgp(\SetS) \}$. Then we run the \PDSTa{} algorithm with input $(\Funcf, \RecDepth, \OptValuep, \Vertv, \varnothing)$ and stop the algorithm once it constructs an out-tree $\SetS$ satisfying $\Funcf(\SetS) \geq \frac{\OptValuep}{2}$. Finally, we return $\SetS$ as the solution to our original instance of \DRCSMa{}. Note that, in order to prove that the returned out-tree $\SetS$ has size $O(\Card^{\frac{1}{\RecDepth}}) \Card$, we would need to invoke the height reduction lemma (Lemma 1) of Calinescu and Zelikosvsky \cite{Calinescu2005}.

\section{Conclusions} \label{sec:con}
We presented a novel polynomial time framework that, for (Directed) \textbf{CSM}, achieves a $\Omega(\frac{\varepsilon^{3}}{{r}^{\varepsilon}})$-approximation for every constant $\varepsilon \in (0, 1]$; and, for \textbf{DRCSM}, achieves a bicriteria $(\Omega(\frac{\delta \varepsilon^{3}}{{r}^{\varepsilon}}), 1 + \delta)$-approximation for every $\delta \in [\frac{1}{\Card}, 1]$ and constant $\varepsilon \in (0, 1]$. This outperforms the state of the art with respect to the size constraint, $\Card$, and the optimal solution radius, $\OptRadius$. As part of our framework, we proposed the algorithms \GR{} and \RAd{} for Directed Rooted \CSMa{}. \GR{} takes a bicriteria $(\AppFactor(\Card), \VioFactor(\Card))$-approximation subroutine and uses it to construct a $(\Omega(\AppFactor(\OptRadius)), O(\VioFactor(\OptRadius)))$-approximate solution. \RAd{} can be used as this subroutine, giving a bicriteria $(\frac{1}{\RecDepth+1}, (\RecDepth+1)^2 \Card^{\frac{1}{\RecDepth}})$-approximation in time $O(\NumVertices^{\RecDepth+1} \Card^{2\RecDepth+2})$.

A potential future direction is to extend our framework to problems with edge or vertex costs, such as \emph{Submodular Tree Orienteering} (edge costs) or \emph{Directed Rooted Submodular Tree} (vertex costs).
Further, for rooted network design problems such as \DRCSMa{} and the aforementioned problems, it is an open problem to show a polynomial time $\Omega(\frac{1}{\Card^{\Error}})$ or $\Omega(\frac{1}{\Radius^{\Error}})$-approximation algorithm that does not violate the size or budget constraint.
Lastly, it would be interesting to prove a stronger approximation hardness result for \CSMa{} than the one already implied by the cardinality constrained problem, as this could inform us on whether our algorithms achieve tight approximation factors with respect to $\Card$. If so, this would further motivate algorithms with beyond-worst-case guarantees.

\begin{acks}
    Philip Cervenjak was supported by the Elizabeth and Vernon Puzey Scholarship, and the Faculty of Engineering and Information Technology.
    Naonori Kakimura was supported by JSPS KAKENHI Grant Numbers JP23K21646, JP22H05001, and JST ERATO Grant Number JPMJER2301, Japan.
    Seeun William Umboh was supported by the Australian Government through the Australian Research Council DP240101353.

    The authors thank the anonymous reviewers for their feedback.
\end{acks}




\bibliographystyle{ACM-Reference-Format} 
\bibliography{main}


\end{document}